\documentclass[12pt,a4paper]{article}
\usepackage{amsfonts}
\usepackage{fancyhdr}
\usepackage{latexsym}
\usepackage{amsmath}
\usepackage{amssymb}
\usepackage{dsfont}
\usepackage{pdflscape}
\usepackage[normalem]{ulem}
\usepackage[font=footnotesize]{caption}
\usepackage{graphicx}
\usepackage{epstopdf}
\usepackage[numbers, square, sort&compress]{natbib} 
\usepackage{bm} 
\usepackage{bbm} 
\usepackage{psfrag}
\usepackage{epsfig}
\usepackage[ansinew]{inputenc}
\usepackage{float}
\usepackage{verbatim}
\usepackage{longtable}
\usepackage{lscape}
\usepackage{tabularx}
\usepackage{multicol}
\usepackage{setspace}
\usepackage{color}
\usepackage{wrapfig}
\usepackage{supertabular}
\usepackage{rotating}
\usepackage{afterpage}
\usepackage{moreverb}       
\usepackage{framed}
\usepackage{caption}
\usepackage{calrsfs}
\usepackage{subfigure}
\usepackage{cases}
\DeclareMathAlphabet{\pazocal}{OMS}{zplm}{m}{n}

\usepackage[bottom,hang]{footmisc}
\usepackage[noprefix]{nomencl}

\usepackage[linktocpage]{hyperref}
\usepackage{url}
\makeatletter
\newcommand\ackname{Acknowledgements}
\if@titlepage
  \newenvironment{acknowledgements}{%
      \titlepage
      \null\vfil
      \@beginparpenalty\@lowpenalty
      \begin{center}%
        \bfseries \ackname
        \@endparpenalty\@M
      \end{center}}%
     {\par\vfil\null\endtitlepage}
\else
  
\fi
\makeatother

\newtheorem{theorem}{Theorem}[section]
\newtheorem{lemma}[theorem]{Lemma}

\newtheorem{corollary}[theorem]{Corollary}
\newtheorem{remark}[theorem]{Remark}

\newtheorem{Def}[theorem]{Definition}

\newenvironment{proof}[1][Proof]{\begin{trivlist}
\item[\hskip \labelsep {\bfseries #1}]}{\end{trivlist}}

\newenvironment{notation}[1][Notation]{\begin{trivlist}
\item[\hskip \labelsep {\bfseries #1}]}{\end{trivlist}}

\newcommand{\halmos}{\quad\hfill\mbox{$\Box$}}

\newcommand*{\IM}{\mathbb{M}}
\newcommand*{\IF}{\mathbb{F}}
\newcommand*{\IE}{\mathbb{E}}
\newcommand*{\IQ}{\mathbb{Q}}
\newcommand*{\IR}{\mathbb{R}}
\newcommand*{\IN}{\mathbb{N}}

\newcommand{\F}{\mathcal{F}}

\newcommand{\MC}{\pazocal{MC}}
\newcommand{\E}{\pazocal{E}}

\numberwithin{equation}{section}

\begin{document}

\begin{center}
{\Large {\bf Portfolio Optimization in Affine Models with Markov Switching}}\\
\singlespacing
{\large Marcos Escobar\footnote{Department of Mathematics, Ryerson University, 350 Victoria St., Toronto ON M5B 2K3 Canada}, Daniela Neykova\footnote{Chair of Mathematical Finance, Technische Universit\"{a}t M\"{u}nchen, Parkring 11, 85748 Garching-Hochbr\"{u}ck, Germany. Email: daniela.neykova@tum.de. (Corresponding author)}, Rudi Zagst\footnote{Chair of Mathematical Finance, Technische Universit\"{a}t M\"{u}nchen, Parkring 11, 85748 Garching-Hochbr\"{u}ck, Germany}}\\
\today
\end{center}
{\bf Abstract}\\
We consider a stochastic factor financial model where the asset price process and the process for the stochastic factor depend on an observable Markov chain and exhibit an affine structure. We are faced with a finite time investment horizon and derive optimal dynamic investment strategies that maximize the investor's expected utility from terminal wealth. To this aim we apply Merton's approach, as we are dealing with an incomplete market. Based on the semimartingale characterization of Markov chains we first derive the HJB equations, which in our case correspond to a system of coupled non-linear PDEs. Exploiting the affine structure of the model, we derive simple expressions for the solution in the case with no leverage, i.e. no correlation between the Brownian motions driving the asset price and the stochastic factor. In the presence of leverage we propose a separable ansatz, which leads to explicit solutions in this case as well. General verification results are also proved. The results are illustrated for the special case of a Markov modulated Heston model.
\section{Introduction}
\label{SecIntr}
In this paper we derive optimal investment strategies by maximizing the expected utility from terminal wealth under a very flexible model, influenced simultaneously by some continuous stochastic factor and a Markov chain.\\
 The utility maximization problem was first stated and solved in continuous time for the Black-Scholes model in \cite{Merton1969} and \cite{Merton1971}. Since then many authors have further developed this theory for more sophisticated and realistic market models. One important extension is the stochastic modeling of the asset returns volatility, as the constant Black-Scholes volatility cannot reflect important empirical observations, such as asymmetric fat-tailed stock return distributions and volatility clustering (see e.g. \cite{Rubinstein1985}, \cite{Engle2001} for detailed empirical studies). Some of the most popular examples of stochastic volatility models include the affine models by Stein\&Stein and Heston, presented in \cite{Stein1991} and \cite{Heston1993}, respectively. Optimal investment rules for a continuous time model with an additional stochastic factor that follows a diffusion process, are derived, e.g., in \cite{Zaripop2001}. \cite{Kraft2005} proves rigorously their validity in the Heston model.\\
Although the incorporation of stochastic volatility makes the asset price modeling more realistic, it does not capture long-term macroeconomic developments. Such fundamental factors can be described by Markov chains, where each state of the Markov chain describes a different market situation, e.g. a crisis or a booming economy. A Markov switching autoregressive model was first introduced and empirically motivated in \cite{Hamilton1989}. A Markov switching Black-Scholes model has been applied for utility maximization e.g. in \cite{Bauerle2004} and \cite{Sotomayor2009}, where the corresponding Hamilton-Jacobi-Bellman (HJB) Equations are stated and optimal controls are derived.\\
Our contribution is the combination of both sources of randomness - a continuous stochastic factor and a Markov chain - in the context of portfolio optimization. To assure the analytical tractability of our model, we assume an affine structure. Such models have been considered recently for derivative pricing. \cite{Elliott2007} derive pricing formulas for volatility swaps and \cite{Mitra2010} and \cite{Papanicolaou2013} apply perturbation methods for option pricing under different Markov switching affine model. Furthermore, empirical confirmations for the relevance of this class of models can be found in \cite{Choi2013}, \cite{Durham2013}.\\
As far as we know, we are the first to derive the optimal portfolio strategy and the corresponding value function for models exhibiting this rich stochastic structure. Our derivations are based on the semimartingale characterization of Markov chains, which is a powerful tool, although not so popular in the literature on Markov chains. It allows us to state the corresponding HJB equations, which in this case result in a system of coupled PDEs. We manage to solve it up to an expectation over the Markov chain, which is easily and quickly calculated. We even allow for instantaneous correlation between the stochastic factor and the asset price process, which is in accordance with empirical observations (see \cite{Engle2001}). As a byproduct of our computations we obtain a version of the famous Feynman Kac Theorem that holds for Markov modulated stochastic processes and can be used to solve systems of coupled PDEs in other applications. Furthermore, we prove a very useful verification result, that reduces the case with Markov switching to the one with deterministic coefficients and is thus very easy to check. Finally we contribute by an analysis of the results from an economic point of view.\\
The remaining of this paper is organized as follows. Section \ref{SecMod} formulates precisely the model we are working with and the optimization problem. Afterwards, Section \ref{SecHJB} gives an overview over the methodology applied for solving it. The obtained solutions are presented in Section \ref{SecSol}. Numerical implementations for the example of the Markov modulated Heston model and analysis of the results can be found in Section \ref{SecHeston}. Section \ref{SecConcl} concludes.

\section{Model and optimization problem}
\label{SecMod}
We consider a continuous-time financial market with a riskless investment opportunity (bank account) and one risky asset. The corresponding price processes are denoted by $\{P_0(t)\}_{t\in[0,T]}$ and $\{P_1(t)\}_{t\in[0,T]}$. The dynamics of these price processes are influenced by a stochastic factor $X$ and an observable continuous-time Markov chain $\MC$. For standard definitions concerning Markov chains and intensity matrices see e.g. \cite{Yin1998}, p.16ff. We assume that the processes follow a special model we call a Markov modulated affine model, which is more precisely defined below:
\begin{Def}[Markov modulated affine model (MMAF)]$\mbox{}$\\
The model is stated on a filtered probability space $(\Omega,\F,\IQ,\IF)$, where $\IQ$ denotes the real world measure. Let $\MC$ be a stationary continuous-time homogeneous Markov chain with finite state space $\E=\{e_1,\ldots,e_l\}$ and intensity matrix $Q=\{q_{ij}\}_{i,j=1,\ldots,l}\in\IR^{l\times l}$, where $e_i$ denotes the $i$-th unit vector in $\IR^l$. Further, denote by $W_P$ and $W_X$ two correlated one-dimensional Brownian motions, independent of the Markov chain. A financial market model is called a Markov modulated affine model if the bank account and the traded asset develop according to the following dynamics : 
\begin{align}
\begin{aligned}
	&\text{d}  P_0(t)=P_0(t)r{ \big(\MC(t)\big)}\text{d}t\\
	&\text{d} P_1(t)=P_1(t)\Big[r\big(\MC(t)\big)+\underbrace{\gamma\big(X(t),\MC(t)\big)\sigma_P\big(X(t),\MC(t)\big)}_{=:\lambda\left(X(t),\MC(t)\right)}\text{d} t\\
&+\sigma_P{ \big(X(t),\MC(t)\big)}\text{d} W_P(t)\Big]\\
	&\text{d}  X(t)=\mu_X(X(t),\MC(t))\text{d}t +\sigma_X(X(t),\MC(t))\text{d} W_X(t)\\
	&\text{d}\langle W_P, W_X\rangle(t)=\rho\text{d}t,
	\end{aligned}
	\label{GeneralModel}
	\end{align}
	with initial values $P_0(0)=p_0$, $P_1(0)=p_1$, $ X(0)=x_0$ and the following specifications called (MMAF):
	\begin{align}
	\begin{aligned}
	\gamma^2(x,e_i)&=\Gamma^{(1)}(e_i)+\Gamma^{(2)}(e_i)x\\
	\mu_X(x,e_i)&=\mu_X^{(1)}(e_i)+\mu_X^{(2)}(e_i)x\\
	\sigma_X^2(x,e_i)&=\Sigma_X^{(1)}(e_i)+\Sigma_X^{(2)}(e_i)x\\
	\rho\gamma(x,e_i)\sigma_X(x,e_i)&=\zeta^{(1)}(e_i)+\zeta^{(2)}(e_i)x,
	\end{aligned}
	\tag{MMAF}
	\label{MMAF}
	\end{align}
for all $(x,e_i)\in D_X\times\E$, where $D_X\subseteq\IR$ denotes the definition set of process $X$ and $\mu_X^{(j)}$,  $\zeta^{(j)}:\E\rightarrow\IR$; $\Gamma^{(j)}$, $\Sigma_X^{(j)}:\E\rightarrow\IR_{\geq0}$, for $j=1,2$, $r$, and $\sigma_P:D_X\times \E\rightarrow\IR$ are deterministic functions. Observe that $\lambda\big(X(t),\MC(t)\big)$ stands for the excess return of the risky asset and $\gamma\big(X(t),\MC(t)\big)$ corresponds to the market price of risk.\\
Furthermore, we denote the filtration generated by the Markov chain by $\IF^{\MC}=\{\F^{\MC}(t)\}_{t\in[0,T]}$.\\
Note that this definition and all results presented in the sequel can be easily extended to additional time-dependence of the parameters.
\end{Def}
\begin{remark}
Note that the last assumption in (\ref{MMAF}) is trivially fulfilled for $\rho=0$, and for $\rho\not=0$ it implies that:
\begin{align}
\sigma_X(x,e_i)=b(e_i)\gamma(x,e_i)\Rightarrow \gamma(x,e_i)=\frac{1}{b(e_i)}\sigma_X(x,e_i),
\label{A4new}
\end{align}
for some deterministic function $b:\E\rightarrow\IR$. So, for the case with correlation the market price of risk should be a multiple of the volatility of the stochastic factor.
\end{remark}
\begin{remark}
From now on throughout the whole paper we denote the set of indices for the states of $\MC$ by $E=\{1,2,\ldots,l\}$ and define process $MC=\{MC(t)\}_{t\geq0}$ by $MC(t):=\sum_{i=1}^{l}i1_{\{\MC(t)=e_i\}}$.
\label{RemMCNotation}
\end{remark}
Observe that the drift and the diffusion term for the stochastic factor, as well as the squared market price of risk (the sharp ratio) are affine in $X$. This property explains the name of the model. We will see in Section \ref{SecHJB} that as a consequence of this, all terms in the corresponding HJB equation have affine structure, which allows for an exponentially affine ansatz for its solution. One important example for this class of models can be obtained by adding Markov switching (MS) to the famous Heston model. This example will be considered in detail in Section \ref{SecHeston}.\\
Note that the two additional sources of randomness - the Markov chain and the stochastic factor - make the market model incomplete. That is why the martingale approach for utility maximization is not applicable and we will rely on the dynamic programming approach.\\
We assume that the investor can observe not only the asset prices but also the value of the stochastic factor $X$ and the state of the Markov chain $\MC$, and makes her investment decision based on all this information. Her strategy is stated in terms of $\pi={\pi(t)}_{t\in[0,T]}$, the relative portion of wealth invested in the risky asset. As we consider only self-financing trading strategies, the SDE for the wealth process $V^{\pi}=\{V^{\pi}(t)\}_{t\in[0,T]}$ generated by strategy $\pi$ (when investment starts at time $0$ with initial wealth $v_0$) is given by:
\begin{align}
\begin{aligned}
\text{d}V^{\pi}(t)=&\underbrace{V^{\pi}(t)\Big[r\big(\MC(t)\big)+\lambda\big(X(t),\MC(t)\big)\pi(t)\Big]}_{=:\mu_V\big(V^{\pi}(t),X(t),\MC(t),\pi(t)\big)}\text{d}t\\
	&+\underbrace{V^{\pi}(t)\pi(t)\sigma_P\big(X(t),\MC(t)\big)}_{=:\sigma_V\big(V^{\pi}(t),X(t),\MC(t),\pi(t)\big)}\text{d}W_P(t)\big]\\
	V^{\pi}(0)=&v_0.
\end{aligned}
\label{SDEV}
	\end{align}
Recalling the exponential form of the solution to this SDE, we can conclude that $V^{\pi}(t)>0$, for all $t\in[0,T]$.\\
The risk preferences of the investor are characterized by the power utility function:
$$U(v)=\frac{v^{\delta}}{\delta}, \delta< 1,\delta\not=0.$$
Observe that the Arrow-Pratt measure of relative risk aversion if given by $$\frac{vU''}{U'}=1-\delta,$$
so the parameter $\delta$ describes the risk aversion of the investor: $\delta\rightarrow 1$ describes a risk-neutral investor and the smaller $\delta$, the more risk-averse the investor. This is confirmed by Figure \ref{fig:utility}, which shows the utility function for different parameter values. It can be seen that the higher $\delta$, the bigger the weight of high wealth levels and the smaller the weight of low wealth values. This means that an investor with a $\delta$ close to 1 is more willing to take risks compared to an investor with a lower $\delta$. This relationship can be recognized even more clearly for negative values for $\delta$, where a very high negative weight is assigned to wealth levels close to zero whereas very high wealth levels improve only marginally the utility of the investor. That is why utility functions with a negative parameter $\delta$ are used to describe very risk averse investors.\\
\begin{figure}[h]
    \centering
    \includegraphics[trim=2cm 6.9cm 2cm 7cm,clip,scale=0.6]{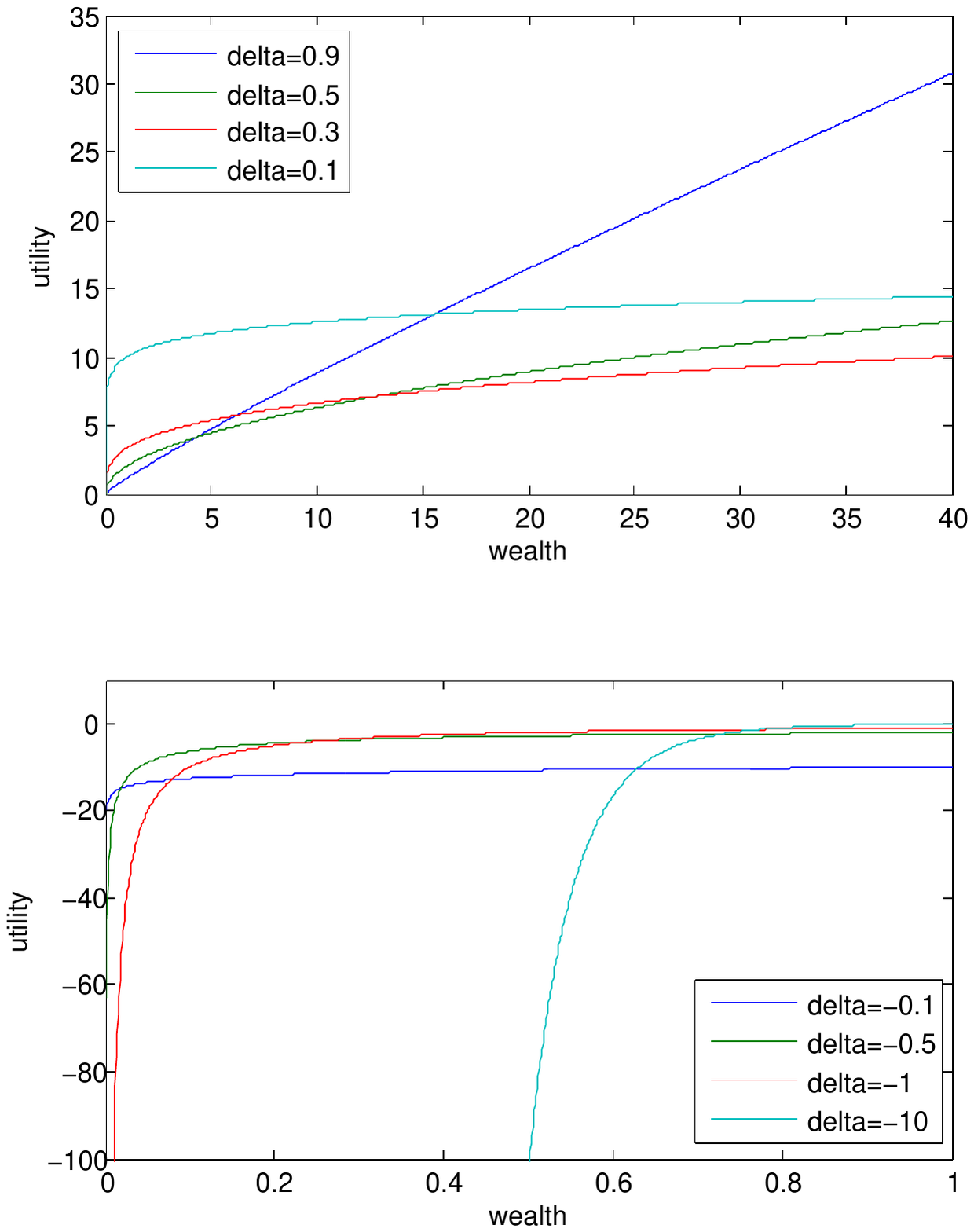}
    \caption{Power utility function $U(v)=\frac{v^{\delta}}{\delta}$ for different parameters $\delta\leq1$.}
    \label{fig:utility}
\end{figure}
The investor aims at maximizing her expected utility from terminal wealth. We state this formally in the following optimization problem for all $(t,v,x,e_i)\in[0,T]\times\IR_{\geq0}\times D_X\times\E$:
\begin{align}
	\begin{aligned}
J^{(t,v,x,e_i)}(\pi)&:=\IE_{\IQ}\Big[U(V^{\pi}(T))|V^{\pi}(t)=v,X(t)=x,\MC(t)=e_i\Big]\\
	\Phi(t,v,x,e_i)&:=\max_{\pi\in\Lambda(t,v)}J^{(t,v,x,e_i)}(\pi),
	\label{optProblem}
	\end{aligned}
	\end{align}
	where the maximal expected utility $\Phi$ is called the value function and $\Lambda(t,v)$ denotes the set of all admissible trading strategies, starting at time $t$ with initial wealth $v$:
	\begin{align*}
	\Lambda(t,v):=\Big\{\pi\Big| \pi(s)\in\IR,V^{\pi}(t)=v,&V^{\pi}(s)\geq0,\forall s\in[t,T],\IE_{\IQ}\Big[U\big(V^{\pi}(T)\big)^{-}|\F_t\Big]<\infty\Big\}.
	\end{align*}
Strictly speaking, one should write $\Lambda(t,v,x,e_i)$, however we omit the remaining arguments for better readability. Observe that in our case the positivity of the wealth process is fulfilled for all self-financing trading strategies. The last condition in the definition of $\Lambda(t,v)$ requires that the negative part of the terminal utility is integrable, excluding strategies that might lead with a positive probability to infinite negative utility. However, if one can show that the value function is finite, then this condition is trivially fulfilled for the corresponding optimal portfolio, as the power utility is either positive or negative on its whole definition set. So, the negative part is either zero or equals the finite value function. Hence, in the considered case, we do not need to check additionally the admissibility of the optimal trading strategies.\\
Observe that now our optimization problem has not only one state variable, as in the classical Merton's optimization problem, but three, corresponding to the  three distinguishable sources of randomness in our model.\\	
Before presenting the solution approach to this optimization problem we refer the reader to Appendix \ref{SecPrel}, where we summarize some general results about Markov chains and Markov modulated It\^o diffusions, which we will need later on.
Now we are equipped with all we need to continue with the solution of Problem (\ref{optProblem}).
\section{HJB approach}
\label{SecHJB}
As we are dealing with an incomplete market, we apply the Hamilton-Jacobi-Bellman (HJB) approach for portfolio optimization. It is briefly described in what follows. The method is based on the so-called Bellman's principle, which states that an optimal control is characterized by the property that whatever the initial action is, the remaining decisions are optimal given the outcome of the first one. This is formally expressed in the following equation:
\begin{align}
\begin{aligned}
\Phi(t,v,x,e_i)=\max_{\pi\in \Lambda(t,v)}\IE_{\IQ}\Big[&\Phi\big(t+h,V^{\pi}(t+h),X(t+h),\MC(t+h)\big)\\
&\Big|V^{\pi}(t)=v,X(t)=x,\MC(t)=e_i\Big].
\end{aligned}
\label{BellmanPrinciple}
\end{align}
Applying the  It\^o's formula for Markov modulated  It\^o diffusions to function $\Phi\big(t+h,V^{\pi}(t+h),X(t+h),\MC(t+h)\big)$, heuristically exchanging the expectation and the integration, and taking the limit for $h\rightarrow0$ in Equation (\ref{BellmanPrinciple}) leads to the following system of coupled HJB equations for the value function $\Phi$, for all $(t,v,x)\in [0,T]\times \IR_{\geq0}\times D_X$: 
\begin{align}
\begin{aligned}
\max_{\pi\in\IR}\{\pazocal{L}(e_i,\pi)\Phi(t,v,{x},e_i)\}&=-\sum_{j=1}^{l}q_{ij}\Phi(t,v,{x},e_j),\forall i\in E\\
	\Phi(T,v,{x},e_i)&=U(v)=\frac{v^{\delta}}{\delta},\forall i\in E,
	\label{HJB}
	\end{aligned}
	\end{align}
where the differential operator $\pazocal{L}(e_i,\pi)$ is given for each $e_i\in \E$ by:
	\begin{align*}
\pazocal{L}(e_i,\pi)\Phi(t,v,{x},e_i):=&{ \Phi_t}+{ \mu_V}(v,x,e_i,\pi){ \Phi_v}+{ \mu_{X}}(x,e_i){ \Phi_{x}} +{ \frac{1}{2}\sigma^2_V}(v,x,e_i,\pi){ \Phi_{vv}}\\
&+{ \frac{1}{2}\sigma_{X}^2}(x,e_i){ \Phi_{xx}}+{ \rho\sigma_{V}}(v,x,e_i,\pi){ \sigma_{X}}(x,e_i){ \Phi_{vx}}.
	\end{align*}
Note that from now on we adopt the following notation for partial derivatives: $F_a=\frac{\partial}{\partial a}F$, $F_{aa}=\frac{\partial^2}{\partial a\partial a}F$, $F_{ab}=\frac{\partial^2}{\partial a\partial b}F$, for any function $F(a,b,\ldots)$.\\
Observe that instead of dealing with only one PDE as in the case with classical  It\^o diffusions, we need to solve a system of coupled partial differential equations.\\
First of all we state the first-order condition for an interior optimum:
\begin{align*}
\frac{\partial}{\partial \pi}\big\{\pazocal{L}(e_i,\pi)\Phi(t,v,{x},e_i)\big\}=&v\lambda(x,e_i)\Phi_v+v^2\pi\sigma^2_P(x,e_i)\Phi_{vv}+\rho v\sigma_x(x,e_i)\sigma_P(x,e_i)\Phi_{vx}=0,
	\end{align*}
where we have inserted the definition of $\mu_V$ and $\sigma_V$ from Equation (\ref{SDEV}). So, the following candidate for the optimal solution can be derived:
	\begin{align}
	\begin{aligned}
	\bar\pi(t)&=\bar{\pi}\big(t,X(t),\MC(t)\big)=-\frac{\lambda\Phi_v+\rho\sigma_X\sigma_P\Phi_{vx}}{V^{\bar\pi}(t)\sigma^2_P \Phi_{vv}}\Big|_{\big(t,X(t),\MC(t)\big)}, \forall t\in[0,T].
	\end{aligned}
	\end{align}
Then, we state the following ansatz for the value function, in accordance with the applied utility function:
	\begin{align}
	\Phi(t,v,x,e_i)=\frac{v^{\delta}}{\delta}f(t,x,e_i), \forall (t,v,x,e_i)\in [0,T]\times \IR_{\geq0}\times D_X\times \E,
	\label{ansatz}
	\end{align}
	for some function $f(t,x,e_i): [0,T]\times D_X\times \E\rightarrow \IR$. So, the candidate for the optimal strategy takes the following form:
	\begin{align}
\bar\pi(t)&=\Big\{\underbrace{\frac{1}{1-\delta}\frac{\lambda}{\sigma^2_P}}_{\text{mean-variance portfolio}}+\underbrace{\frac{1}{1-\delta}\rho\frac{\sigma_X}{\sigma_P}\frac{f_x}{f}}_{\text{hedging term}}\Big\}\Big|_{\big(t,X(t),\MC(t)\big)}\label{optStrat1}\\
&=\frac{1}{1-\delta}\frac{1}{\sigma_P}\Big\{\gamma+\rho\sigma_X\frac{f_x}{f}\Big\}\Big|_{\big(t,X(t),\MC(t)\big)}.
	\label{optStrat2}
		\end{align}
The first part of $\bar{\pi}$ exhibits the same structure as the optimal portfolio in the Black-Scholes model without a stochastic factor and corresponds (up to a constant factor) to the solution of a classical mean-variance problem, that is why we call it the mean-variance portfolio. The second term can be interpreted as a hedge against the additional risk coming from the stochastic factor. Note that it disappears if $\rho=0$, because in this case the stochastic factor cannot be hedged with the traded asset. The representation of $\bar\pi$ in (\ref{optStrat2}) allows for an intuitive interpretation of the optimal portfolio. One can recognize that the higher the volatility term of the risky asset, the lower the absolute value of its portfolio weight. This means that increasing risk reduces the optimal exposure to the asset. On the other side, increasing market price of risk influences positively the investment in $P_1$, as the asset becomes more attractive to the investor. For the interpretation of the second summand in (\ref{optStrat2}) additional knowledge about function $f$ is necessary. A detailed analysis in the case of the Markov modulated Heston model can be found in Section \ref{SubHestonNumRho}. Furthermore, observe that $\bar{\pi}$ does not depend on the wealth process $V^{\bar{\pi}}$.\\ 
Inserting Equations (\ref{ansatz}) and (\ref{optStrat2}) in the HJB equation (\ref{HJB}) leads to the following reduced system of PDEs for the unknown functions $f(t,x,e_i)$, for all $(t,x)\in [0,T]\times D_X$:
		\begin{align}
\begin{aligned}
	&{f_t}(t,x,{ e_i})+{ f}(t,x,{ e_i})\delta \Big\{r(e_i)+ \frac{1}{2}\frac{1}{1-\delta}\gamma^2(x,e_i)\Big\}+{ f_x}(t,x,{ e_i})\Big\{\mu_X(x,e_i)\\ &+\frac{\delta}{1-\delta}\rho\sigma_X(x,e_i)\gamma(x,e_i)\Big\}+\frac{1}{2}{ f_{xx}}(t,x,{ e_i})\sigma_X^2(x,e_i)\\
	&+\frac{1}{2}\frac{{ f^2_x}(t,x,{ e_i})}{{ f}(t,x,{ e_i})}\frac{\delta}{1-\delta}\rho^2\sigma_X^2(x,e_i)={ -\sum_{j=1}^{l}q_{ij}f(t,x,{ e_j})}, f(T,x,{ e_i})=1,\forall i\in E.
\end{aligned}
\label{PDEfgeneral}
	\end{align}
Solving this system analytically is challenging and the coupling of the equations complicates numerical procedures. However, in the next section we present a list of special models where by conveniently placing the dependence on the Markov chain one can obtain simple expressions that can be easily computed numerically.\\
Observe that only finding a solution of the HJB equation is not enough to show the optimality of the resulting strategy. A verification theorem needs to be proved in order to show that all technical assumptions are indeed satisfied. The verification results are also contained in what follows.
\section{Explicit solutions}
\label{SecSol}
For the case of no leverage (i.e. $\rho=0$), we provide the solution in the general Model (\ref{GeneralModel}) up to a very simple expectation only over the probability measure of the Markov chain, which can be computed very efficiently (see Theorem \ref{ThExplSolNoRho}). For a special case we even derive an explicit solution up to a deterministic Riccati ODE (see Theorem \ref{ThExplSolNoRhoSpecCase}). For the case where the Brownian motions exhibit instantaneous correlation, closed-form solutions are derived for a relevant choice of the parameters (see Theorem \ref{ThVerifRho}). What is more, in Theorem \ref{ThVerifNoRho} we state a verification result that is very easy to validate.\\
Before proceeding with these solutions we need to solve the optimization problem in an auxiliary model, where the parameters switch according to a deterministic step function. The precise definition of the model and the solution to the problem, as well as the needed verification result are given in Section \ref{SubsTimedep}.
\subsection{Time-dependent affine models as an intermediate step}
\label{SubsTimedep}
We will start with the results for the general time-dependent model (Section \ref{SubsubsGenTime}) and then we will present a special case, where the verification result can be shown in a very convenient way (Section \ref{SubsubsSpecTime}).
\subsubsection{General time-dependent affine model}
\label{SubsubsGenTime}
Consider the following stochastic factor model:
\begin{align}
\begin{aligned}
	&\text{d}  P^m_0(t)=P^m_0(t)r{ \big(m(t)\big)}\text{d}t\\
	&\text{d} P^m_1(t)=P^m_1(t)\Big[r\big(m(t)\big)+\underbrace{\gamma\big(X^m(t),m(t)\big)\sigma_P\big(X^m(t),m(t)\big)}_{=\lambda(X^m(t),m(t))}\text{d} t\\
&+\sigma_P{ \big(X^m(t),m(t)\big)}\text{d} W_P(t)\Big]\\
	&\text{d}  X^m(t)=\mu_X(X^m(t),m(t))\text{d}t +\sigma_X(X^m(t),m(t))\text{d} W_X(t)\\
	&\text{d}\langle W_P, W_X\rangle(t)=\rho\text{d}t,
	\end{aligned}
	\label{TimedepModel}
	\end{align}
with initial values $P^m_0(0)=p_0$, $P^m_1(0)=p_1$ and $X^m(0)=x_0$. Note that $r$, $\gamma$, $\lambda$, $\sigma_P$, $\mu_X$ and $\sigma_X$ are the same deterministic function as in the definition of Model (\ref{GeneralModel}) and $m$ is a deterministic piecewise constant function with values in $\E=\{e_1,\ldots,e_l\}$. Let us denote the set of all these functions by $\IM$. More precisely, $m$ is given by:
	\begin{align}
	m(t):=\begin{cases}
	m_0 & t\in [0,t_1)\\
	m_1 & t\in [t_1,t_2)\\
	\vdots\\
	m_K & t\in [t_{K},T),
	\end{cases}
	\label{m}
		\end{align}
	where $t_j$, $j=1,\ldots,K$ with $t_0:=0<t_1<t_2<\ldots<t_K\leq T=:t_{K+1}$ denote the jump times of $m$, $K$ is the number of jumps till time $T$ and $m_j\in\E$, $j=1,\ldots,K$ are the corresponding states. We can interpret this model as Model (\ref{GeneralModel}) conditioned on an arbitrary but fixed path of the Markov chain. That is why we call it the corresponding time-dependent model to Model (\ref{GeneralModel}) or the time-dependent model induced by $m$ and vice versa, Model (\ref{GeneralModel}) is the Markov switching model corresponding to Model (\ref{TimedepModel}).\\
As before, $P^m_0$ denotes the price of the riskless investment, $P^m_1$ is the price process of the risky asset and $V^{m,\pi}$ stays for the wealth process corresponding to the self-financing strategy $\pi$. We consider again a power utility function $U(v)=\frac{v^{\delta}}{\delta}$ and the following optimization problem:
\begin{align}
	\begin{aligned}
J^{(t,v,x,m)}(\pi)&:=\IE_{\IQ}\Big[U\big(V^{m,\pi}(T)\big)\Big|V^{m,\pi}(t)=v,X^m(t)=x\Big]\\
	\Phi^m(t,v,x)&:=\max_{\pi\in\Lambda^m(t,v)}J^{(t,v,x,m)}(\pi),
	\end{aligned}
	\end{align}
	where
	\begin{align*}
	\Lambda^m(t,v):=\Big\{\pi\Big|\pi(s)\in\IR\wedge V^{m,\pi}(s)\geq0,\forall s\in[t,T]\wedge\IE_{\IQ}\Big[U\big(V^{m,\pi}(T)\big)^{-}\Big|\F_t\Big]<\infty\Big\}.
	\end{align*}
	The HJB equation in this case takes the following form:
\begin{align}
\begin{aligned}
\max_{\pi\in\IR}\{\pazocal{L}(m(t),\pi)\Phi^m(t,v,{x})\}&=0, 
	\Phi^m(T,v,{x})=\frac{v^{\delta}}{\delta},
	\end{aligned}
	\label{HJBTimeDep}
	\end{align}
where the differential operator $\pazocal{L}$ is defined for all $t\in[0,T]$ via:
	\begin{align*}
&\pazocal{L}(m(t),\pi)\Phi^m(t,v,{x}):={ \Phi^m_t}+{ \mu_V}(v,x,m(t),\pi){ \Phi^m_v}+{ \mu_{X}}(x,m(t)){ \Phi^m_{x}} \\
&+{ \frac{1}{2}\sigma^2_V}(v,x,m(t),\pi){ \Phi^m_{vv}}+{ \frac{1}{2}\sigma_{X}^2}(x,m(t)){ \Phi^m_{xx}}+{ \rho\sigma_{V}}(v,x,m(t),\pi){ \sigma_{X}}(x,m(t)){ \Phi^m_{vx}}.
	\end{align*}
Analogously as before we use the first-order condition for an interior maximum $\frac{\partial}{\partial \pi}\big\{\pazocal{L}(m(t),\pi)\Phi^m(t,v,{x})\big\}=0$ to obtain a candidate for the optimal investment strategy:
	\begin{align}
	\begin{aligned}
	\bar{\pi}^{m}(t)&=-\frac{\lambda\Phi^m_v+\rho\sigma_X\sigma_P\Phi^m_{vx}}{V^{m,\pi}(t)\sigma^2_P \Phi^m_{vv}}\Big|_{\big(t,X^m(t),m(t)\big)}.
	\end{aligned}
	\end{align}
By the ansatz: 
\begin{align}
{ \Phi^m(t,v,x)=\frac{v^{\delta}}{\delta}f^m(t,x)},
\label{AnsatzTimedep}
\end{align}
we simplify $\bar{\pi}^{m}$ to:
	\begin{align}	
\bar{\pi}^{m}(t)&=\frac{1}{1-\delta}\Big\{\frac{\lambda}{\sigma^2_P}+\rho\frac{\sigma_X}{\sigma_P}\frac{f^m_x}{f^m}\Big\}\Big|_{\big(t,X^m(t),m(t)\big)}=\frac{1}{1-\delta}\frac{1}{\sigma_P}\Big\{\gamma+\rho\sigma_X\frac{f^m_x}{f^m}\Big\}\Big|_{\big(t,X^m(t),m(t)\big)}.
	\label{PiTimeDep}
		\end{align}
Substitution of Equations (\ref{AnsatzTimedep}) and (\ref{PiTimeDep}) in the HJB equation (\ref{HJBTimeDep}) leads to the following PDE in $t$ and $x$ for function $f^m$:
\begin{align}
\begin{aligned}
	&f^m_t(t,x)+f^m(t,x)\underbrace{\delta \Big\{r(m(t))+ \frac{1}{2}\frac{1}{1-\delta}\gamma^2(x,m(t))\Big\}}_{=:g(x,m(t))}\\ &+f^m_x(t,x)\underbrace{\Big\{\mu_X(x,m(t))+\frac{\delta}{1-\delta}\rho\sigma_X(x,m(t))\gamma(x,m(t))\Big\}}_{=:\tilde{\mu}_X(x,m(t))}\\
	&+\frac{1}{2}{ f^m_{xx}}(t,x)\sigma_X^2(x,m(t))+\frac{1}{2}\frac{{ \big(f^m_x}(t,x)\big)^2}{{f^m}(t,x)}\frac{\delta}{1-\delta}\rho^2\sigma_X^2(x,m(t))=0,f^m(T,x)=1.
	\end{aligned}
	\label{PDETimeDepfm}
	\end{align}
 Note that in this case we obtain one PDE and not a system of PDEs as in the case with Markov switching. As in \cite{Zaripop2001} we apply the following transformation to get rid of the nonlinear term: 
 \begin{align*}
 h^m(t,x):=\big(f^m(t,x)\big)^{\frac{1}{\vartheta}},
 \end{align*}
 for $\vartheta=\frac{1-\delta}{1-\delta+\delta\rho^2}$. Then $h^m$ should solve the PDE below: \begin{align}
\begin{aligned}
&h^m_t(t,x)+h^m(t,x)\frac{1}{\vartheta}g(x,m(t))+h^m_x(t,x)\tilde{\mu}_X(x,m(t))+\frac{1}{2}{h^m_{xx}}(t,x)\sigma_X^2(x,m(t))=0,\\
	&h^m(T,x)=1.
\end{aligned}
\label{PDEhTimedep}
	\end{align}
Observe that $\vartheta=1$ for $\rho=0$. Under some integrability conditions we can apply Corollary \ref{CorFK} to obtain the following probabilistic representation for the solution of the PDE from above:
	\begin{align}
	\begin{aligned}
h^m(t,x)=&\IE\Big[\exp\Big\{\int_t^T\frac{1}{\vartheta}g\big(\tilde{X}^m(s),m(s)\big)\Big\}\Big|\tilde{X}^m(t)=x\Big]\\
=&\exp\Big\{\int_t^T\frac{1}{\vartheta}\delta \Big\{r(m(s))+ \frac{1}{2}\frac{1}{1-\delta}\Gamma^{(1)}(m(s))\Big\}\text{d}s\Big\}\\
&\IE\Big[\exp\Big\{\int_t^T \frac{1}{\vartheta} \frac{1}{2}\frac{\delta}{1-\delta}\Gamma^{(2)}(m(s))\tilde X^m(s)\text{d}s\Big\}\Big|\tilde{X}^m(t)=x\Big],
\end{aligned}
\label{ProbReprTimeDeph}
	\end{align}
	where the dynamics of process $\tilde{X}^m$ are given for $t\in[0,T]$ by the following SDE:
	\begin{align*}
	\text{d}\tilde{X}^m(t)=\tilde{\mu}_X\big(\tilde{X}^m(t),m(t)\big)\text{d}t+\sigma_X\big(\tilde{X}(t), m(t)\big)\text{d}W_{{X}}(s).
	\end{align*}
Observe that when the expectation in Representation (\ref{ProbReprTimeDeph}) is known explicitly in a closed form, we do not need to check the assumptions of Corollary \ref{CorFK}, as we can just plug in the solution and verify that it is indeed the solution to the HJB equation. In what follows we will characterize this solution up to a Riccati ODE.\\
First recall the affine definition of our model, which implies:
	\begin{align*}
	\tilde{\mu}_X (x,e )&=\mu_X^{(1)} (e )+\frac{\delta}{1-\delta}\zeta^{(1)} (e ) + \{\mu_X^{(2)} (e )+\frac{\delta}{1-\delta}\zeta^{(2)} (e ) \}x.
	\end{align*}
This affine structure allows for an affine ansatz for function $h^m$. More precisely, we assume that:
\begin{align*}
h^m(t,x)=\exp\Big\{\int_t^T\frac{1}{\vartheta}\delta \Big\{r(m(s))+ \frac{1}{2}\frac{1}{1-\delta}\Gamma^{(1)}(m(s))\Big\}\text{d}s\Big\}\exp\{A^m(t)+B^m(t)x\}.
\end{align*}
Inserting this ansatz in Equation (\ref{PDEhTimedep}) and equating the coefficients in front of $x^0$ and $x^1$ leads to the following two ODEs for $A^m$ and $B^m$:
\begin{align}
&B^m_t(t)+\frac{1}{2}\Sigma_X^{(2)}\big(m(t)\big)\big(B^m(t)\big)^2+\big\{\frac{\delta}{1-\delta}\zeta^{(2)}\big(m(t)\big)+\mu_X^{(2)}\big(m(t)\big)\big\}B^m(t)\notag\\
&+\frac{1}{2}\frac{1}{\vartheta}\frac{\delta}{1-\delta}\Gamma^{(2)}\big(m(t)\big)=0,B^m(T)=0
\label{Riccati}\\
&A^m_t(t)+\frac{1}{2}\Sigma_X^{(1)}\big(m(t)\big)\big(B^m(t)\big)^2+\big\{\frac{\delta}{1-\delta}\zeta^{(1)}\big(m(t)\big)+\mu_X^{(1)}\big(m(t)\big)\big\}B^m(t)=0,\notag\\
& A^m(T)=0. 
\label{Integr}
\end{align}
So in order to obtain the solution of the HJB equation we only need to solve the Riccati Equation (\ref{Riccati}) and do the integration in Equation (\ref{Integr}):
\begin{align*}
A^m(t)&=\int_t^T\frac{1}{2}\Sigma_X^{(1)}\big(m(s)\big)\big(B^m(s)\big)^2+\big\{\frac{\delta}{1-\delta}\zeta^{(1)}\big(m(s)\big)+\mu_X^{(1)}\big(m(s)\big)\big\}B^m(s)\text{d}s.
\end{align*}
Observe that the parameters in these equations are time-dependent but piecewise constant. So we are looking for continuous, piecewise continuously differentiable solutions for $A^m$ and $B^m$. If we know the solution for constant parameters and inhomogeneous terminal condition a recursive solution method with finitely many steps is possible. A result on the solvability of Equation (\ref{Riccati}) and the probabilistic representation of its solution is given in Proposition 5.1. in \cite{Kraft2005} and summarized in Lemma \ref{jointCharF} in Section \ref{SecHeston}, where we apply it to construct stepwise a solution for the example of the Heston model.\\
The next theorem summarizes what we just derived.
\begin{theorem}[HJB solution in the time-dependent affine model]$\mbox{}$\\
Assume that Equation (\ref{Riccati}) admits a unique continuous, piecewise continuously differentiable solution $B^m$. Then the solution of the HJB equation (\ref{HJBTimeDep}) is given by:
\begin{align}
&\Phi^{m}(t,v,x)=\frac{v^{\delta}}{\delta}\IE\Big[\exp\Big\{\int_t^T\frac{1}{\vartheta}g\big(\tilde{X}^m(s),m(s)\big)\text{d}s\Big\}\Big|\tilde{X}^m(t)=x\Big]^{\vartheta}\notag\\
&=\frac{v^{\delta}}{\delta}\exp\Big\{\int_t^T\delta \Big\{r(m(s))+ \frac{1}{2}\frac{1}{1-\delta}\Gamma^{(1)}(m(s))\Big\}\text{d}s\Big\}\exp\{{\vartheta} A^m(t)+{\vartheta} B^m(t)x\}\\
&=\frac{v^{\delta}}{\delta}\exp\Big\{\int_t^T\delta r\big(m(s)\big)+\frac{1}{2}\frac{\delta}{1-\delta}\Gamma^{(1)}\big(m(s)\big)+\frac{1}{2}\vartheta\Sigma_X^{(1)}\big(m(s)\big)\big(B^m(s)\big)^2\notag\\
&+\vartheta\big\{\frac{\delta}{1-\delta}\zeta^{(1)}\big(m(s)\big)+\mu_X^{(1)}\big(m(s)\big)\big\}B^m(s)\text{d}s\Big\}\exp\{\vartheta B^m(t)x\}.
\label{EqSolTimeD}
\end{align}
The corresponding candidate for the optimal portfolio has the following form (see Equation (\ref{PiTimeDep})):
\begin{align}
\begin{aligned}
\bar{\pi}^{m}(t)&=\frac{1}{1-\delta}\Big\{\frac{\lambda}{\sigma^2_P}+\rho\frac{\sigma_X}{\sigma_P}\vartheta B^m\Big\}\Big|_{(t,X^m(t),m(t))}=\frac{1}{1-\delta}\frac{1}{\sigma_P}\Big\{\gamma+\rho\sigma_X\vartheta B^m\Big\}\Big|_{(t,X^m(t),m(t))}.
\end{aligned}
\label{OptimStratTimeDepGen}
\end{align}
Trivially, $\Phi^{m}\in\mathcal{C}^{1,2,2}\big([t_n,t_{n+1})\times\IR_{\geq0}\times D_X\big)$ for all $n\in\{0,1,2,\ldots,K\}$ and it is continuous on the whole interval $[0,T]$.
\label{ThSolTimeD}
\end{theorem}
What is still left to be done is to show that the solution of the HJB equation is indeed the value function of our optimization problem. The next theorem summarizes a sufficient set of conditions for this. It holds not only for affine models but for general models with a stochastic factor and time-dependent piecewise constant parameters.
\begin{theorem}[Verification result in the time-dependent affine model]$\mbox{}$\\
Consider a real-valued function $\Phi^m(t,v,x):[0,T]\times \IR_{\geq0}\times D_X\rightarrow \IR$ and assume that:
\begin{enumerate}
\renewcommand{\labelenumi}{\roman{enumi})}
\item $\Phi^m\in\mathcal{C}^{1,2,2}\big([t_i,t_{i+1})\times \IR_{\geq0}\times D_X\big)$ for all $i=0,\ldots,k$,
\item $\Phi^m\in\mathcal{C}\big([0,T]\times \IR_{\geq0}\times D_X\big)$, 
\item $\Phi^m$ satisfies the following PDE, defined piecewise for all $i=1,\ldots,K$:
\begin{align*}
\max_{\pi\in\IR}\{\pazocal{L}(m_i,\pi)\Phi^m(t,v,{x})\}&=0,\text{ for all }(t,v,x)\in[t_i,t_{i+1})\times \IR_{\geq0}\times D_X\\
	\Phi^m(T,v,{x})&=\frac{v^{\delta}}{\delta},
\end{align*}
where the maximum on the left hand side is obtained at $$\bar{\pi}^m=-\frac{\lambda\Phi^m_v+\rho\sigma_X\sigma_P\Phi^m_{vx}}{V^{m,\pi}(t)\sigma^2_P \Phi^m_{vv}}\Big|_{\big(t,X^m(t),m(t)\big)}.$$
\end{enumerate}
If $\Phi^m$ is positive then:
\begin{align*}
\IE\big[U\big(V^{m,\pi}(T)\big)\big|V^{m,\pi}(t)=v,X^m(t)=x\big]\leq \Phi^m(t,v,x),
\end{align*}
for all $(t,v,x)\in[0,T]\times \IR_{\geq0}\times D_X$ and all admissible portfolio strategies $\pi$.\\
Now consider again a not necessarily positive function $\Phi^m$ satisfying Conditions i), ii) and iii). Assume additionally that:
\begin{itemize}
\item[iv)] $\big\{\Phi^m\big(t,V^{m,\bar{\pi}^m}(t),X^m(t)\big)\big\}_{t\in[0,T]}$ is a martingale.
\end{itemize}
Then:
\begin{align*}
\IE\big[U\big(V^{m,\bar{\pi}^m}(T)\big)\big|V^{m,\bar{\pi}^m}(t)=v,X^m(t)=x\big]= \Phi^m(t,v,x),
\end{align*}
and 
\begin{align*}
\IE\big[U\big(V^{m,\pi}(T)\big)\big|V^{m,\pi}(t)=v,X^m(t)=x\big]\leq \Phi^m(t,v,x),
\end{align*}
for all $(t,v,x)\in[0,T]\times \IR_{\geq0}\times D_X$ and all admissible portfolio strategies $\pi$.
\label{ThVerifTimeDGen}
\end{theorem}
The proof is given in Appendix \ref{AppSecSol}.\\
So, the conditions in Theorem \ref{ThVerifTimeDGen} need to be checked for the parameters of the special model of interest. This might be quite laborious. That is why we present in the following section a special case of the general time-dependent model, for which the verification result can be shown easily.
\subsubsection{Special time-dependent affine model with $\Sigma_X^{(1)}=\Gamma^{(1)}=\zeta^{(1)}=0$}
\label{SubsubsSpecTime}
For this section we additionally assume for Model (\ref{TimedepModel}) that:
\begin{align*}
\Sigma_X^{(1)}=\Gamma^{(1)}=\zeta^{(1)}=0.
\end{align*}
Then the verification result follows easily from one general statement about exponentials of affine processes, given in Corollary 3.4 from \cite{Kallsen2010}. For convenience, this result is summarized in Lemma \ref{LemmaKallsen} in Appendix \ref{AppSecSol}. Now we apply it for the verification result. 

\begin{theorem}[Verification result for $\Sigma_X^{(1)}=\Gamma^{(1)}=\zeta^{(1)}=0$]$\mbox{}$\\
Let the assumptions of Theorem \ref{ThSolTimeD} hold true and $\Sigma_X^{(1)}=\Gamma^{(1)}=\zeta^{(1)}=0$. Then the optimal portfolio strategy is as  given in Equation (\ref{OptimStratTimeDepGen}) and $\Phi^m$ as given by Equation (\ref{EqSolTimeD}) is the corresponding value function.

\label{ThVerifTimeD}
\end{theorem}
The proof is given in Appendix \ref{AppSecSol}
\begin{remark}
An alternative verification Theorem for the Heston model is presented in \cite{Kraft2005}. It relies on the specific form of the solution in this example. However, it is very technical and requires some restrictions on the model parameters.
\end{remark}
To summarize, when we consider a special example for a time-dependent model with $\Sigma_X^{(1)}=\Gamma^{(1)}=\zeta^{(1)}=0$, we just need to find continuous, piecewise differentiable functions $A^m$ and $B^m$ that solve Equations (\ref{Riccati}) and (\ref{Integr}), respectively. Then we can apply the just proved theorems and derive the optimal solution. In Section \ref{SecHeston} we will do this for a Heston-type model.\\ 
After deriving all necessary results for the time-dependent model, we can proceed with the Markov modulated models. We will consider separately the cases with and without correlation between the Brownian motions driving the stock price and the stochastic factor, as the presence of correlation complicates significantly the derivations.
\subsection{Markov modulated affine models with independent Brownian motions}
\label{SubsNoRho}
In this section we set $\rho=0$ in Model (\ref{GeneralModel}). Observe that although the Brownian motions do not exhibit instantaneous correlation, the two processes are additionally correlated by the joint Markov chain. After considering the general case in Section \ref{SubsunNoRhoGeneral}, we will present in Section \ref{SubsunNoRhoSepa} a special separable case, where the solution can be further simplified.
\subsubsection{General Markov-modulated affine model (MMAF) with $\rho=0$}
\label{SubsunNoRhoGeneral}
In what follows we first make a guess for the HJB solution based on Corollary \ref{CorFK} and then show that it is indeed the value function to our optimization problem.\\
Observe that the system of PDEs (\ref{PDEfgeneral}) for function $f$ simplifies to:
\begin{align}
\begin{aligned}
	&{f_t}(t,x,{ e_i})+{ f}(t,x,{ e_i})\underbrace{\delta \Big\{r(e_i)+ \frac{1}{2}\frac{1}{1-\delta}\gamma^2(x,e_i)\Big\}}_{=g(x,e_i)}+{ f_x}(t,x,{ e_i})\mu_X(x,e_i)\\ 
	&+\frac{1}{2}{ f_{xx}}(t,x,{ e_i})\sigma_X^2(x,e_i)={ -\sum_{j=1}^{l}q_{ij}f(t,x,{ e_j})}, f(T,x,{ e_i})=1, \forall\ i\in E.
\end{aligned}
\label{PDEfNoRho}
	\end{align}
So, if the conditions of Corollary \ref{CorFK} are fulfilled for process $X$ and function $g(x,e_i)$ then the solution to System (\ref{PDEfNoRho}) is given by:
\begin{align*}
f(t,x,e_i)=&\IE\Big[\exp\Big\{\int_t^T\delta\Big(r\big(\MC(s)\big)+ \frac{1}{2}\frac{1}{1-\delta}\gamma^2\big(X(s),\MC(s)\big)\Big)\text{d}s\Big\}\\
&\Big|X(t)=x,\MC(t)=e_i\Big].
\end{align*}
Further, to ease the readability of the derivations below, we introduce the following notation for any $n$-dimensional process $Z$: $Z_{t,z}=(Z^1_{t,z},\ldots,Z^n_{t,z})^{\prime}$ denotes the process $Z$ started at time $t$ in point $z=(z_1,\ldots,z_n)^{\prime}$. Now rewrite the previous equation in the new notation and transform it by the tower rule for conditional expectations as follows:
\begin{align*}
&f(t,x,e_i)=\IE\Big[\exp\Big\{\int_t^T\delta\Big(r\big(\MC_{t,e_i}(s)\big)+ \frac{1}{2}\frac{1}{1-\delta}\gamma^2\big(X_{t,x,e_i}(s),\MC_{t,e_i}(s)\big)\Big)\text{d}s\Big\}\Big]\\
&=\IE\Big[\IE\Big[\exp\Big\{\int_t^T\delta\Big(r\big(\MC_{t,e_i}(s)\big)+ \frac{1}{2}\frac{1}{1-\delta}\gamma^2\big(X_{t,x,e_i}(s),\MC_{t,e_i}(s)\big)\Big)\text{d}s\Big\}|\F_T^{\MC}\Big]\Big]\\
&=\IE\Big[f^{\MC_{t,e_i}}(t,x)\Big]=\IE\Big[f^{\MC}(t,x)\Big|\MC(t)=e_i\Big],
\end{align*}
for all $(t,x,e_i)\in[0,T]\times D_X\times\E$, where $f^m$ denotes for all $m\in\IM$ the solution of System (\ref{PDETimeDepfm}) in the time-dependent model induced by $m$, thus $f^{\MC}(t,x)$ is an $\F_T^{\MC}$-measurable random variable. Observe that here we have used the independence of $\MC$ and $W_X$.
So the candidate for the value function in the considered model has the following form:
\begin{align*}
\Phi(t,v,x,e_i)&=\frac{v^{\delta}}{\delta}f(t,x,e_i)=\frac{v^{\delta}}{\delta}\IE\Big[f^{\MC_{t,e_i}}(t,x)\Big]=\IE\Big[\frac{v^{\delta}}{\delta}f^{\MC_{t,e_i}}(t,x)\Big]\\
&=\IE[\Phi^{\MC_{t,e_i}}(t,v,x)]=\IE[\Phi^{\MC}(t,v,x)|\MC(t)=e_i],
\end{align*}
where for each $m\in\IM$, $\Phi^m$ denotes the value function in the time-dependent model. This insight inspired us to show a general verification result, which reduces the case with Markov switching to the case with deterministic time-dependent model parameters and thus can be easily checked. One important advantage of this verification result is that we do not need the statement that our candidate value function solves the HJB equations but we show directly that it is indeed the value function of our problem. The next theorem provides this result.
\begin{theorem}[Verification result with independent Brownian motions]$\mbox{}$\\
Assume that for all $m\in\IM$ the value function in the time-dependent model induced by $m$ is given by $\Phi^m(t,v,x):[0,T]\times\IR_{\geq 0}\times D_X\rightarrow \IR$ and the optimal investment strategy $\bar{\pi}^m$ depends on the current level of $m$ and the current values of the stochastic processes, but not on the whole path of $m$ on $[0,T]$, i.e. it is given by
\begin{align*}
\bar{\pi}^m(t)=p(t,m(t),V^{m,\bar{\pi}^m}(t),X^m(t)),
\end{align*}
for some function $p:[0,T]\times\E\times\IR_{\geq0}\times D_X\rightarrow\IR$. Then, the value function in the corresponding Markov modulated model is given by:
\begin{align*}
\Phi(t,v,x,e_i)=\IE[\Phi^{\MC}(t,v,x)|\MC(t)=e_i],
\end{align*}
and the optimal investment strategy is $\bar{\pi}(t)=p(t,\MC(t),V^{\bar{\pi}}(t),X(t))$.
\label{ThVerifNoRho}
\end{theorem}
\begin{proof}
Let $\bar\pi^m$ and $\bar\pi$ be as defined in the theorem above. First note the following equality for an arbitrary but fix $m\in\IM$:
$$\big(X^{m}_{t,x},V^{m,\bar{\pi}^{m}}_{t,v,x},\bar{\pi}^{m}_{t,v,x}\big)\equiv\big(X_{t,x,m(t)},V^{\bar{\pi}}_{t,v,x,m(t)},\bar{\pi}_{t,v,x,m(t)}\big)\big|\{\MC(s)=m(s),\forall s\in[t,T]\}.$$
This implies that:
$$\big(X^{\MC}_{t,x,e_i},V^{\MC,\bar{\pi}^{\MC}}_{t,v,x,e_i},\bar{\pi}^{\MC}_{t,v,x,e_i}\big)\equiv\big(X_{t,x,e_i},V^{\bar{\pi}}_{t,v,x,e_i},\bar{\pi}_{t,v,x,e_i}\big).$$
Pay attention that this is true because the optimal strategy $\bar{\pi}^m$ in the time-dependent model depends only on the current value of the stepwise function $m$ and not on its whole path on $[t,T]$ and because of the independence between the Brownian motions and the Markov chain. By applying this observation and the tower rule for conditional expectations we obtain:
\begin{align*}
\Phi(t,v,x,e_i)&=\IE\big[\Phi^{\MC}(t,v,x)\big|\MC(t)=e_i\big]=\IE\big[\Phi^{\MC_{t,e_i}}(t,v,x)\big]\\ 
&=\IE\Big[\IE\big[\Phi^{\MC_{t,e_i}}(t,v,x)\big|\F_T^{\MC_{t,e_i}}\big]\Big]=\IE\Big[\IE\big[\frac{(V_{t,v,x,e_i}^{\MC,\bar{\pi}^{\MC}}(T))^{\delta}}{\delta}\big|\F_T^{\MC_{t,e_i}}\big]\Big]\\
&=\IE\Big[\IE\big[\frac{\big(V_{t,v,x,e_i}^{\bar{\pi}}(T)\big)^{\delta}}{\delta}\big|\mathcal{F}^{\MC_{t,e_i}}_T\big]\Big]=\IE\Big[\frac{\big(V_{t,v,x,e_i}^{\bar{\pi}}(T)\big)^{\delta}}{\delta}\Big].
\end{align*}
So, function $\Phi$ expresses indeed the expected utility corresponding to strategy $\bar{\pi}$. \\
To obtain the optimality of $\bar{\pi}$ consider an arbitrary admissible $\pi\in\Lambda(t,v,x,e_i)$ and note that $\pi|\{\MC(s)=m(s),\forall s\in[t,T]\}\in \Lambda^{m}(t,v,x),\forall m\in\IM$, thus $\pi\in\Lambda^{\MC_{t,e_i}}(t,v,x)$. This means that conditional on the whole path of the Markov chain $\pi$ is admissible for the time-dependent model, where function $m$ corresponds to the path of the Markov chain. Thus, using the optimality of $\bar{\pi}^{\MC}$, one can compute
\begin{align*}
\Phi(t,v,x,e_i)&=\IE\Big[\IE\big[\frac{\big(V_{t,v,x,e_i}^{\MC,\bar{\pi}^{\MC}}(T)\big)^{\delta}}{\delta}\big|\F_{T}^{\MC_{t,e_i}}\big]\Big]\\
&\geq \IE\Big[\IE\big[\frac{\big(V_{t,v,x,e_i}^{\MC,\pi}(T)\big)^{\delta}}{\delta}\big|\F_{T}^{\MC_{t,e_i}}\big]\Big]=\IE\Big[\frac{\big(V_{t,v,x,e_i}^{\pi}(T)\big)^{\delta}}{\delta}\Big].
\end{align*}
$\mbox{}$\halmos
\end{proof}
\begin{remark}
Observe that for the proof above the assumption $\rho=0$ is not needed. Thus, the statement holds also for the general model with $\rho\not=0$.
\label{RemarkVerifRho}
\end{remark}
We apply this theorem to the solution for the time-dependent model obtained in Section \ref{SubsTimedep} and provide the result in the following theorem.
\begin{theorem}[Solutions with independent Brownian motions]$\mbox{}$\\
Consider Model (\ref{GeneralModel}) and set $\rho=0$. Assume that the value function in the corresponding time-dependent model is given by:
\begin{align*} 
\Phi^{m}(t,v,x)&=\frac{v^{\delta}}{\delta}\IE\Big[\exp\Big\{\int_t^Tg(X^m(s),m(s))\text{d}s\Big\}\Big|X^m(t)=x\Big]=:\frac{v^{\delta}}{\delta}f^m(t,x)
\end{align*}
and the optimal investment strategy by:
\begin{align*}
\bar{\pi}^m(t)=\frac{1}{1-\delta}\frac{\lambda}{\sigma_P^2}\Big|_{(X^m(t),m(t))}=\frac{1}{1-\delta}\frac{\gamma}{\sigma_P}\Big|_{(X^m(t),m(t))}.
\end{align*}
Then, the value function in the Markov modulated model has the following form:
\begin{align} 
\begin{aligned}
\Phi(t,v,x,e_i)=&\IE[\Phi^{\MC}(t,v,x)|\MC(t)=e_i]\\
=&\frac{v^{\delta}}{\delta}\IE\Big[\exp\Big\{\int_t^Tg(X(s),\MC(s))\text{d}s\Big\}\Big|X(t)=x,\MC(t)=e_i\Big]\\
&=:\frac{v^{\delta}}{\delta}f(t,x,e_i),
\label{PhiNoRhoProb}
\end{aligned}
\end{align}
and the optimal portfolio strategy is 
\begin{align*}
\bar{\pi}(t)=\underbrace{\frac{1}{1-\delta}\frac{\lambda}{\sigma_P^2}\Big|_{(X(t),\MC(t))}}_{\text{mean-variance portfolio}}=\frac{1}{1-\delta}\frac{\gamma}{\sigma_P}\Big|_{(X(t),\MC(t))}.
\end{align*}
Note that as mentioned before, in the case of $\rho=0$ the optimal strategy consists only on the mean-variance portfolio.
\label{ThExplSolNoRho}
\end{theorem}
\begin{proof}
Verify that the optimal portfolio in the time-dependent model does not depend on the whole path of function $m$ and apply Theorem \ref{ThVerifNoRho}.\\
$\mbox{}$\halmos
\end{proof}
After deriving the general solution we are now interested in simplifying the probabilistic representation from Theorem \ref{ThExplSolNoRho}. To this aim, we consider a special model presented in what follows.
\subsubsection{Separable Markov-modulated affine model (SMMAF) with $\rho=0$}
\label{SubsunNoRhoSepa}
 In this section we assume separability of function $f(t,x,e_i)$ in $x$ and $e_i$, i.e. we consider the following ansatz:
\begin{align*}
f(t,x,e_i)=\bar\xi(t,e_i)\exp\big\{B(t)x\big\},
\end{align*}
for some functions $\bar\xi:[0,T]\times\E\rightarrow \IR$ and $B:[0,T]\rightarrow \IR$. In order to be able to separate the PDE for $f$, we need to assume that the coefficients that appear in front of the stochastic factor $X$ do not depend on the Markov chain. More precisely, the following structure for the model parameters, called (SMMAF), is considered:
\begin{align}
\begin{aligned}
\gamma^2(x,e_i)&=\Gamma^{(1)}(e_i)+\bar\Gamma^{(2)}x\\
\mu_X(x,e_i)&=\mu_X^{(1)}(e_i)+\bar\mu_X^{(2)}x\\
\sigma^2_X(x,e_i)&=\Sigma_X^{(1)}(e_i)+\bar\Sigma_X^{(2)}x
\end{aligned}
\tag{SMMAF}
\label{SMMAF}
\end{align}
where $\bar\Gamma^{(2)},\bar\Sigma_X^{(2)}\in\IR_{\geq0}$, $\bar\mu_X^{(2)}\in\IR$ are some constants. The probabilistic representation from Theorem \ref{ThExplSolNoRho} still cannot be simplified, as process $X$ depends on the Markov chain. However direct substitution of the ansatz and the parameter specifications in PDE (\ref{PDEfNoRho}) turns to be useful. It leads to the following system of PDEs for function $\bar\xi(t,e)$, for all $t\in[0,T]$:
\begin{align}
\begin{aligned}
&\bar\xi_t(t,e_i)+\bar\xi(t,e_i)\underbrace{\big[\delta r(e_i)+\frac{1}{2}\frac{\delta}{1-\delta}\Gamma^{(1)}(e_i)+B(t)\mu_X^{(1)}(e_i)+\frac{1}{2}B^2(t)\Sigma_X^{(1)}(e_i)\big]}_{=:w(t,e_i)}\\
&=-\sum_{j=1}^{l}q_{ij}\bar\xi(t,e_j),\ \ \bar\xi(T,e_i)=1,\forall i\in E,
\end{aligned}
\label{PDExi}
\end{align}
and the following ODE for $B(t)$:
\begin{align}
B_t(t)+\frac{1}{2}B^2(t)\bar\Sigma_X^{(2)}+B(t)\bar\mu_X^{(2)}+\frac{1}{2}\frac{\delta}{1-\delta}\bar\Gamma^{(2)}=0,B(T)=0\label{RiccatiNoRho}.
\end{align}
Equation (\ref{RiccatiNoRho}) can be solved by Lemma \ref{jointCharF}.\\
In general, the solution for System (\ref{PDExi}) cannot be derived in a closed form, however, one can approximate it by the so-called Magnus exponential series, for details see \cite{Magnus1954}. Applying a numerical solving scheme is also possible. Alternatively, Corollary \ref{CorFKLS} provides the following probabilistic representation for  function $\bar\xi(t,e)$:
\begin{align}
\bar\xi(t,e_i)=\IE\Big[\exp\Big\{\int_{t}^{T}w(s,\MC(s))\text{d}s\Big\}\Big|\MC(t)=e_i\Big],\forall (t,e_i)\in[0,T]\times\E,
\label{XiProb}
\end{align}
as $w(t,e)$ and $\frac{\partial}{\partial t}w(t,e)$ are continuous in $t$.
In the following theorem we summarize what we have just derived:
\begin{theorem}[Solution and verification in the separable case with $\rho=0$]$\mbox{}$\\
Let the model specifications be given by \ref{SMMAF} and assume that Equation (\ref{RiccatiNoRho}) admits a unique differentiable solution $B$. Then the solution of the corresponding HJB equation is given by: 
\begin{align}
\begin{aligned}
\Phi(t,v,x,e_i)=&\frac{v^{\delta}}{\delta}\IE\Big[\exp\Big\{\int_{t}^{T}w(s,\MC(s))\text{d}s\Big\}\Big|\MC(t)=e_i\Big]\exp\{B(t)x\}\\
=&\frac{v^{\delta}}{\delta}\bar\xi(t,e_i)\exp\{B(t)x\}, \forall (t,v,x,e_i)\in[0,T]\times\IR_{\geq0}\times D_X\times \E,
\end{aligned}
\label{PhiNoRho}
\end{align}
where $w(t,e_i)=\delta r(e_i)+\frac{1}{2}\frac{\delta}{1-\delta}\Gamma^{(1)}(e_i)+B(t)\mu_X^{(1)}(e_i)+\frac{1}{2}B^2(t)\Sigma_X^{(1)}(e_i)$, for all $(t,e_i)\in[0,T]\times\E$. Note that $\Phi\in\mathcal{C}^{1,2,2}$ for all $e_i\in\E$.\\
If $\bar\Sigma_X^{(2)}\not=0$, $\bar\mu_X^{(2)}<0$, $\frac{\delta}{1-\delta}\bar\Gamma^{(2)}<\frac{\big(\bar\mu_X^{(2)}\big)^2}{\bar\Sigma_X^{(2)}}$ and $0\leq\frac{a-\bar\mu_X^{(2)}}{\bar\Sigma_X^{(2)}}$, for $a:=\sqrt{\big(\bar\mu_X^{(2)}\big)^2-\frac{\delta}{1-\delta}\bar\Gamma^{(2)} \bar\Sigma_X^{(2)}}$, then function $B$ is given by\footnote{As we will see in Section \ref{SecHeston}, for the important example where $X$ follows a CIR process it holds $\bar\mu_X^{(2)}<0$.}:
\begin{align*}
B(t)=\begin{cases}\frac{-c(-\bar\mu_X^{(2)}+a)\exp\{-a(T-t)\}-\bar\mu_X^{(2)}-a}{\bar\Sigma_X^{(2)}(1-c\exp\{-a(T-t)\})}&\text{for } 0<\frac{a-\bar\mu_X^{(2)}}{\bar\Sigma_X^{(2)}}\\0&\text{for } 0=\frac{a-\bar\mu_X^{(2)}}{\bar\Sigma_X^{(2)}}\end{cases},
\end{align*}
with $c:=\frac{-\bar\mu_X^{(2)}-a}{-\bar\mu_X^{(2)}+a}$.\\
Further, consider an arbitrary but fixed path $m$ of the Markov chain and consider the time-dependent model associated with $m$. Then, the solution $\Phi^m$ to its HJB equation is given by:
\begin{align*}
\Phi^{m}(t,v,x)&=\frac{v^{\delta}}{\delta}\exp\Big\{\int_t^Tw\big(s,m(s)\big)\text{d}s\Big\}\exp\big\{B(t)x\big\}.
\end{align*}
Now assume that for all $m$, $\Phi^{m}$ is indeed the value function in the corresponding time-dependent model. Then, function $\Phi$ is the value function for the model with Markov switching and the optimal portfolio is:
\begin{align*}
\bar{\pi}(t)\frac{1}{1-\delta}\frac{\lambda}{\sigma_P^2}\Big|_{(X(t),\MC(t))}=\frac{1}{1-\delta}\frac{\gamma}{\sigma_P}\Big|_{(X(t),\MC(t))}.
\end{align*}
\label{ThExplSolNoRhoSpecCase}
\end{theorem}
\begin{proof}
The first statement follows directly form the derivations above the theorem and Corollary \ref{CorFKLS}, which can be applied as $w(t,e_i)$ and $\frac{\partial}{\partial t}w(t,e_i)$ are continuous in $t$ for all $e_i\in\E$. For the explicit representation of $B$ consider additionally Lemma \ref{jointCharF}.\\
For the HJB solution in the time-dependent model observe that $B(t)$ solves Equation (\ref{Riccati}) for the considered model specification and  apply Theorem \ref{ThSolTimeD}. Remembering that for $\rho=0$, $\vartheta=1$ and $\zeta^{(1)}=0$, it is easily verified that the term in the integral in Equation (\ref{EqSolTimeD}) corresponds to $w(s,m(s))$. So, $\Phi^m$ solves the time-dependent HJB equation and by assumption it is its value function. Observe that for function $\Phi$ it holds:
\begin{align*}
\Phi(t,v,x,e_i)=\IE[\Phi^{\MC}(t,v,x)|\MC(t)=e_i].
\end{align*}
Further, the optimal strategy in the time-dependent model is given by:
\begin{align*}
\bar{\pi}^m=\frac{1}{1-\delta}\frac{\lambda}{\sigma_P^2}\Big|_{(X^m(t),m(t))}=\frac{1}{1-\delta}\frac{\gamma}{\sigma_P}\Big|_{(X^m(t),m(t))}.
\end{align*}
As it does not depend on the whole path of $m$, but only on its current level, we can apply Theorem \ref{ThVerifNoRho}. The statement follows directly.\\
$\mbox{}$\halmos
\end{proof}
Before illustrating these results by the example of the Heston model, we present the solution in the case, where we allow for instantaneous correlation between the stock price process and the stochastic factor.
\subsection{Markov modulated affine models with leverage}
\label{SubsRho}
In this section we consider the general Model (\ref{GeneralModel}), so function $f$ is characterized by System (\ref{PDEfgeneral}). We were not able to find its solutions in general, mainly because of the nonlinear term. Unfortunately a transformation like the one in Section \ref{SubsTimedep} does not work in the most general case, because here we have a system of coupled PDEs. What is more, Theorem \ref{ThVerifNoRho} cannot be applied in this case, as in general $\bar{\pi}^m$ depends on the whole path of $m$. The key to find a solution in this case is to assume a separable exponential ansatz:
\begin{align}
f(t,x,e_i)=\xi(t,e_i)\exp\big\{D(t)x\big\},
\label{ansatz2Rho}
\end{align}
for some functions $\xi:[0,T]\times\E\rightarrow\IR$, $D :[0,T]\rightarrow\IR$, which of course implies certain restrictions on the model parameters. More precisely, to be able to separate the considered PDE in $e_i$ and $x$ we assume the following parameter specifications, called ($\text{SMMAF}^{\rho}$):
\begin{align}
\begin{aligned}
\gamma^2(x,e_i)&=\Gamma^{(1)}(e_i)+\bar\Gamma^{(2)}x\\
\mu_X(x,e_i)&=\mu_X^{(1)}(e_i)+\bar\mu_X^{(2)}x\\
\sigma^2_X(x,e_i)&=b^2(\Gamma^{(1)}(e_i)+\bar\Gamma^{(2)}x)
\end{aligned}
\tag{$\text{SMMAF}^{\rho}$}
\label{SMMAFrho}
\end{align}
for some real constants $\bar\Gamma^{(2)}\in\IR_{\geq0}$, $b\in\IR_{>0}$ and $\bar\mu_X^{(2)}\in\IR$. Note that this parameter setting is a special case of the specifications in System (\ref{SMMAF}), where we additionally need to assume that $\sigma_X$ is a constant multiple of $\gamma$, or equivalently $\gamma(x,e_i)=\frac{1}{b}\sigma_X(x,e_i)$, so the market price of risk is proportional to the the volatility of the stochastic factor.\\
By inserting (\ref{ansatz2Rho}) in PDE (\ref{PDEfgeneral}) we obtain the following equations for functions $\xi(t,e)$ and $D(t)$, for all $t\in[0,T]$:
\begin{align}
&\xi(t,e_i)\underbrace{\Big[\delta r(e_i)+\frac{1}{2}\frac{\delta}{1-\delta}\Gamma^{(1)}(e_i)+D\mu_X^{(1)}(e_i)+D\frac{\delta}{1-\delta}\rho b\Gamma^{(1)}(e_i)+\frac{1}{2}D^2\frac{b^2}{\vartheta}\Gamma^{(1)}(e_i)\Big]}_{=:\upsilon(t,e_i)}\notag\\
&+\xi_t(t,e_i)=-\sum_{j=1}^{l}q_{ij}\xi(t,e_j),\xi(T,e_i)=1,\forall i\in E\\
&D_t+\frac{1}{2}\frac{\delta}{1-\delta}\bar\Gamma^{(2)}+D\big[\bar\mu_X^{(2)}+\frac{\delta}{1-\delta}\rho b\bar\Gamma^{(2)}\big]+\frac{1}{2}D^2\frac{b^2}{\vartheta}\bar\Gamma^{(2)}=0, D(T)=0,\label{RiccatiRho}
\end{align}
where $\vartheta=\frac{1-\delta}{1-\delta+\delta\rho^2}$. Analogously as before, we obtain from Corollary \ref{CorFKLS} the following expression for $\xi$:
\begin{align}
\xi(t,e_i)&=\IE\Big[\exp\Big\{\int_{t}^{T}\upsilon(s,\MC(s))\text{d}s\Big\}\Big|\MC(t)=e_i\Big], \forall e_i\in\E.
\label{xiRho}
\end{align}
Observe that it is continuous and differentiable w.r.t $t$. As mentioned in the previous section, $\xi$ can be computed in general either by numerically solving the system of ODEs, by  Magnus series or by a Monte Carlo Simulation for (\ref{xiRho}). For $D$ we need to solve a Riccati ODE with constant parameters (see Lemma \ref{jointCharF}). Let us summarize what we have just computed and provide the needed verification result. 
\begin{theorem}[Solution and verification with leverage]$\mbox{}$\\
Consider the model specification as in (\ref{SMMAFrho}) and assume that Equation (\ref{RiccatiRho}) pocesses a unique differentiable solution $D$. Then the solution of the corresponding HJB equation is given by: 
\begin{align}
\begin{aligned}
\Phi(t,v,x,e_i)=&\frac{v^{\delta}}{\delta}\IE\Big[\exp\Big\{\int_{t}^{T}\upsilon(s,\MC(s))\text{d}s\Big\}\Big|\MC(t)=e_i\Big]\exp\{D(t)x\}\\
=&\frac{v^{\delta}}{\delta}\xi(t,e_i)\exp\{D(t)x\}, \forall (t,v,x,e_i)\in[0,T]\times\IR_{\geq0}\times D_X\times \E,
\end{aligned}
\label{PhiRho}
\end{align}
where function $\xi$ is given as in Equation (\ref{xiRho}). Note that $\Phi\in\mathcal{C}^{1,2,2}$ for all $e_i\in\E$.\\
If $b\bar\Gamma^{(2)}\not=0$, $\bar\mu_X^{(2)}+\frac{\delta}{1-\delta}\rho b \bar\Gamma^{(2)}<0$, $\frac{\delta}{1-\delta}\bar\Gamma^{(2)}<\frac{\vartheta (\bar\mu_X^{(2)}+\frac{\delta}{1-\delta}\rho b\bar\Gamma^{(2)})^2}{b^2\bar\Gamma^{(2)}}$ and $0\leq\frac{\vartheta(a-\bar\mu_X^{(2)}-\frac{\delta}{1-\delta}\rho b\bar\Gamma^{(2)})}{b^2\bar\Gamma^{(2)}}$, for $a:=\sqrt{(\bar\mu_X^{(2)}+\frac{\delta}{1-\delta}\rho b\bar\Gamma^{(2)})^2-\frac{\delta}{1-\delta}\frac{b^2}{\vartheta}(\bar\Gamma^{(2)})^2}$, then function $D$ is given by:
\begin{align*}
D(t)=\begin{cases}\frac{\vartheta\big(-c(-\bar\mu_X^{(2)}-\frac{\delta}{1-\delta}\rho b\bar\Gamma^{(2)}+a)\exp\{-a(T-t)\}-\bar\mu_X^{(2)}-\frac{\delta}{1-\delta}\rho b\bar\Gamma^{(2)}-a\big)}{b^2\bar\Gamma^{(2)}(1-c\exp\{-a(T-t)\})}&\text{for } 0<\frac{\vartheta(a-\bar\mu_X^{(2)}-\frac{\delta}{1-\delta}\rho b\bar\Gamma^{(2)})}{b^2\bar\Gamma^{(2)}}\\0&\text{for } 0=\frac{\vartheta(a-\bar\mu_X^{(2)}-\frac{\delta}{1-\delta}\rho b\bar\Gamma^{(2)})}{b^2\bar\Gamma^{(2)}}\end{cases},
\end{align*}
with $c:=\frac{-\bar\mu_X^{(2)}-\frac{\delta}{1-\delta}\rho b\bar\Gamma^{(2)}-a}{-\bar\mu_X^{(2)}-\frac{\delta}{1-\delta}\rho b\bar\Gamma^{(2)}+a}$.\\
Further, consider an arbitrary but fixed path $m$ of the Markov chain and consider the time-dependent model associated with $m$. Then, the solution $\Phi^m$ to its HJB equation is given by:
\begin{align}
\Phi^{m}(t,v,x)&=\frac{v^{\delta}}{\delta}\exp\Big\{\int_t^T\upsilon\big(s,m(s)\big)\text{d}s\Big\}\exp\big\{D(t)x\big\}.
\label{HelpPhiTimeD}
\end{align}
Assume that for all $m$, $\Phi^{m}$ is indeed the value function in the corresponding time-dependent model. Then, function $\Phi$ is the value function for the model with Markov switching and the optimal portfolio is given by:
\begin{align*}
\bar{\pi}&=\Big\{\underbrace{\frac{1}{1-\delta}\frac{\lambda}{\sigma^2_P}}_{\text{mean-var. portf.}}+\underbrace{\frac{1}{1-\delta}\rho\frac{\sigma_X}{\sigma_P}D}_{\text{hedging term}}\Big\}\Big|_{(t,(X(t),\MC(t))}=\frac{1}{1-\delta}\frac{1}{\sigma_P}\Big\{\gamma+\rho\sigma_XD(t)\Big\}\Big|_{(t,(X(t),\MC(t))}.
\end{align*}
\label{ThVerifRho}
\end{theorem}
\begin{proof}
For the first statement consider the derivations above and Corollary \ref{CorFKLS}. The explicit expression for function $D$ follows by Lemma \ref{jointCharF}. For the proof of Equation (\ref{HelpPhiTimeD}), observe that $B(t):=\frac{1}{\vartheta}D(t)$ solves Equation \ref{Riccati} for the considered model specification and  apply Theorem \ref{ThSolTimeD}. It is easily verified that the term in the integral in Equation (\ref{EqSolTimeD}) corresponds to $\upsilon(s,m(s))$. So, $\Phi^m$ solves the time-dependent HJB equation and by assumption it is its value function. Observe that for function $\Phi$ it holds:
\begin{align*}
\Phi(t,v,x,e_i)=\IE[\Phi^{\MC}(t,v,x)|\MC(t)=e_i].
\end{align*}
Further, the optimal strategy in the time-dependent model is given by:
\begin{align*}
\bar{\pi}^m&=\frac{1}{1-\delta}\Big\{\frac{\lambda}{\sigma^2_P}+\rho\frac{\sigma_X}{\sigma_P}D\Big\}\Big|_{(t,(X^m(t),m(t))}=\frac{1}{1-\delta}\frac{1}{\sigma_P}\Big\{\gamma+\rho\sigma_XD\Big\}\Big|_{(t,(X^m(t),m(t))}.
\end{align*}
As it does not depend on the whole path of $m$, but only on its current level, we can apply Theorem \ref{ThVerifNoRho} and \ref{RemarkVerifRho}. The statement follows directly.\\
$\mbox{}$\halmos
\end{proof}
 
\section{Example: Markov modulated Heston model (MMH)}
\label{SecHeston}
In this section we apply the derived results to the famous Heston model, where the stochastic factor follows a mean reverting CIR process and is interpreted as the stochastic volatility of the asset price process. The original model was introduced in \cite{Heston1993}. Optimal portfolios under the original Heston model are derived in \cite{Kraft2005} and \cite{Kallsen2008}.\\
In what follows we extend this framework to Markov switching parameters. More precisely, we consider the following model:
\begin{align}
\begin{aligned}
&\text{d}  P_0(t)= P_0(t)r{\big(\MC(t)\big)}\text{d}t\\
&	\text{d} P_1(t)=P_1(t)\Big[r\big(\MC(t)\big)+\hat{\lambda}\big(\MC(t)\big)X(t)\text{d} t+\nu\big(\MC(t)\big)\sqrt{X(t)}\text{d} W_P(t)\Big]\\
	&\text{d}{X}(t)=\kappa\big(\MC(t)\big)(\theta\big(\MC(t)\big)-X(t))\text{d}t+\chi\big(\MC(t)\big)\sqrt{X(t)}\text{d}W_X(t)\\
&	\text{d}\langle W_P, W_X\rangle(t)=\rho,
\label{MCHeston}
\end{aligned}
\end{align}
with initial values $P_0(0)=p_0$, $P_1(0)=p_1$ and $X(0)=x_0$ and $r,\hat{\lambda},\nu,\kappa,\theta,\chi:\E\rightarrow\IR$ being deterministic functions with $\kappa(e_i),\theta(e_i),\chi(e_i)>0$, for all $e_i\in\E$. Furthermore, it is assumed that $2\kappa(e_i)\theta(e_i)\geq\chi^2(e_i)$, for all $e_i\in\E$, in order to assure the positivity of process $X^m$. Observe that this framework corresponds to the following parameter specifications in terms of the notation from Model (\ref{GeneralModel}) and (\ref{MMAF}):
\begin{align}
\begin{aligned}
\sigma_P(x,e_i)&=\nu(e_i)\sqrt{x}&&\\
\gamma(x,e_i)&=\frac{\hat{\lambda}(e_i)}{\nu(e_i)}\sqrt{x}&&\Rightarrow \Gamma^{(1)}(e_i)=0,\ \  \Gamma^{(2)}(e_i)=\frac{\hat{\lambda}^2(e_i)}{\nu^2(e_i)}\\
\mu_X(x,e_i)&=\kappa(e_i)\big(\theta(e_i)-x\big)&&\Rightarrow \mu_X^{(1)}(e_i)=\kappa(e_i)\theta(e_i),\ \  \mu_X^{(2)}(e_i)=-\kappa(e_i)\\
\sigma_X(x,e_i)&=\chi(e_i)\sqrt{x}&&\Rightarrow \Sigma_X^{(1)}(e_i)=0,\ \  \Sigma_X^{(2)}(e_i)=\chi^2(e_i)\\
\rho\gamma(x,e_i)\sigma_X(x,e_i)&=\rho\frac{\hat{\lambda}(e_i)}{\nu(e_i)}\chi(e_i)x&&\Rightarrow \zeta^{(1)}(e_i)=0,\ \ \zeta^{(2)}(e_i)=\rho\frac{\hat{\lambda}(e_i)}{\nu(e_i)}\chi(e_i),
\end{aligned}
\tag{MMH}
\label{MMH}
\end{align}
for all $(t,e_i)\in[0,T]\times \E$. A similar model is used for pricing of volatility swaps in \cite{Elliott2007}.\\
As before, we first derive the optimal portfolio strategy and the value function in the corresponding time-dependent model (Section \ref{SubsHestonTimedep}). Afterwards, in Section \ref{SubsHestonNoRhoModel}, we present the results for the Markov modulated model without correlation between the Brownian motion driving the risky asset price process and the one for the volatility. Section \ref{SubsHestonRhoModel} deals with the case with correlation. 
\subsection{Time-dependent Heston model}
\label{SubsHestonTimedep}
The time-dependent model is stated as follows:
\begin{align}
\begin{aligned}
&\text{d}  P_0(t)= P_0(t)r{\big(m(t)\big)}\text{d}t, \\
&\text{d} P_1(t)=P_1(t)\Big[r\big(m(t)\big)+\hat{\lambda}\big(m(t)\big)X^m(t)\text{d} t+\nu\big(m(t)\big)\sqrt{X^m(t)}\text{d} W_P(t)\Big], \\
	&\text{d}{X^m}(t)=\kappa\big(m(t)\big)(\theta\big(m(t)\big)-X^m(t))\text{d}t+\chi\big(m(t)\big)\sqrt{X^m(t)}\text{d}W_X(t), \\
&\text{d}\langle W_P, W_X\rangle(t)=\rho\text{ d}t,
\end{aligned}
\label{ModelHestTimeDep}
\end{align}
with initial values $P_0(0)=p_0$, $P_1(0)=p_1$ and $X(0)=x_0$, where function $m$ is defined as in Equation (\ref{m}). A similar model is presented and motivated in the context of calibration and derivatives pricing in \cite{Elices2007} and \cite{Mikhailov2003}.\\
From Equation (\ref{PDETimeDepfm}) in Section \ref{SubsTimedep} we know that the drift $\tilde{\mu}_X$ of the modified process $\tilde{X}^m$ is given by:
\begin{align*}
\tilde{\mu}_X(\tilde{X}^m(t),m(t))&={\kappa(m(t))}\big({\theta}(m(t))-\tilde{X}^m(t)\big)+\frac{\delta}{1-\delta}\rho \frac{\chi(m(t))\hat{\lambda}(m(t))}{\nu(m(t))}\tilde{X}^m(t)\\
&=:\tilde{\kappa}(m(t))\big(\tilde{\theta}(m(t))-\tilde{X}^m(t)\big).
\end{align*}
Further recall that we are interested in solving Equations (\ref{Riccati}) and (\ref{Integr}) in order to find the value function of our optimization problem. Furthermore, by Theorem \ref{ThSolTimeD} we have the following probabilistic representation:
\begin{align*}
h^m(t,x)=\IE\Big[\exp\Big\{\int_t^T\frac{1}{\vartheta} g(s,\tilde{X}^m(s),m(s))\text{d}s\Big\}\Big|\tilde{X}^m(t)=x\Big].
\end{align*}
Naturally it leads to the same ODEs (\ref{Riccati}) and (\ref{Integr}), however in the context of the time-dependent parameters it is more convenient to work with the expectation. The following theorem brings together the probabilistic and the ODE approach and delivers the main tool to calculate the solution explicitly. It is cited from \cite{Kraft2005}, Proposition 5.1.
\begin{lemma}$\mbox{}$\\
Define process $X$ by the following SDE:
\begin{align*}
\text{d}{X}(t)=\kappa(\theta-X(t))\text{d}t+\chi\sqrt{X(t)}\text{d}W_1(t),
\end{align*}
where $\kappa,\theta,\chi\in\IR_{>0}$ and $W_1$ is a standard Brownian motion. For $\beta\leq\frac{\kappa^2}{2\chi^2}$ and $\alpha\leq\frac{\kappa+a}{\chi^2}$, where $a:=\sqrt{\kappa^2-2\beta\chi^2}$, the following function is well-defined:
\begin{align*}
\varphi^{\alpha,\beta}(t,T,x):=\IE\Big[\exp\Big\{\alpha X(T)+\beta\int_t^TX(s)\text{d}s\Big\}\Big|X(t)=x\Big].
\end{align*}
More precisely, it is given by: 
\begin{align*}
\varphi^{\alpha,\beta}(t,T,x)=\exp\big\{A^{\alpha,\beta}(T-t)+B^{\alpha,\beta}(T-t)x\big\},
\end{align*}
where for fixed $T>0$ functions $A^{\alpha,\beta}(\tau)$ and $B^{\alpha,\beta}(\tau)$ are real-valued, continuously differentiable on $[0,T]$ and satisfy the following system of ODEs:
\begin{align*}
&-B^{\alpha,\beta}_{\tau}(\tau)+\frac{1}{2}\chi^2 \big( B^{\alpha,\beta}(\tau)\big)^2-\kappa B^{\alpha,\beta}(\tau)+\beta=0,  B^{\alpha,\beta}(0)=\alpha\\
&- A^{\alpha,\beta}_{\tau}(\tau)+\kappa\theta  B^{\alpha,\beta}(\tau)=0,  A^{\alpha,\beta}(0)=0.
\end{align*}
For $\beta<\frac{\kappa^2}{2\chi^2}$ and $\alpha<\frac{\kappa+a}{\chi^2}$ they are given by:
\begin{align*}
 A^{\alpha,\beta}(\tau)&=\frac{\kappa\theta(\kappa-a)}{\chi^2}\tau-\frac{2\kappa\theta}{\chi^2}\ln\Big\{\frac{1-c\exp(-a\tau)}{1-c}\Big\}\\
 B^{\alpha,\beta}(\tau)&=\frac{-c(\kappa+a)\exp(-a\tau)+\kappa-a}{\chi^2\big(1-c\exp(-a\tau)\big)},
\end{align*}
where $c:=\frac{-\alpha\chi^2+\kappa-a}{-\alpha\chi^2+\kappa+a}$. For $\beta\leq\frac{\kappa^2}{2\chi^2}$ and $\alpha=\frac{\kappa+a}{\chi^2}$ we obtain:
\begin{align*}
& A^{\alpha,\beta}(\tau)=\kappa\theta\frac{\kappa+a}{\chi^2}\tau,\ \  B^{\alpha,\beta}(\tau)=\frac{\kappa+a}{\chi^2}.
\end{align*}
Observe that functions $\tilde A^{\alpha,\beta}(t):=A^{\alpha,\beta}(T-t)$ and $\tilde B^{\alpha,\beta}(t):=B^{\alpha,\beta}(T-t)$ solve the following system of ODEs:
\begin{align*}
&\tilde B^{\alpha,\beta}_{t}(t)+\frac{1}{2}\chi^2 \big( \tilde B^{\alpha,\beta}(t)\big)^2-\kappa \tilde B^{\alpha,\beta}(t)+\beta=0,  \tilde B^{\alpha,\beta}(T)=\alpha\\
& \tilde A^{\alpha,\beta}_{t}(t)+\kappa\theta  \tilde B^{\alpha,\beta}(t)=0,  \tilde A^{\alpha,\beta}(T)=0.
\end{align*}
\label{jointCharF}
\end{lemma}
In the following lemma we summarize some properties of functions $A^{\alpha,\beta}$ and $B^{\alpha,\beta}$, which
we will need for the computation of the value function F m and for the qualitative analysis
of the optimal portfolio.
\begin{lemma}[Properties of the characteristic function]$\mbox{}$\\
Consider the notation from Lemma \ref{jointCharF} and assume that $\beta<\frac{\kappa^2}{2\chi^2}$ and $\alpha<\frac{\kappa+a}{\chi^2}$. Then it holds that:
\begin{enumerate}
\renewcommand{\labelenumi}{\roman{enumi})}
\item $B^{\alpha,\beta}(\tau)$ is monotone in $\tau$.
\item $\lim_{\tau\downarrow 0}B^{\alpha,\beta}(\tau)=\alpha$.
\item $\lim_{\tau\uparrow \infty}B^{\alpha,\beta}(\tau)=\frac{\kappa-a}{\chi^2}\begin{cases}
<0,&\text{ for }\beta<0\\=0,&\text{ for }\beta=0\\>0,&\text{ for }\beta>0\end{cases}$.
\item $\frac{\partial}{\partial\tau}A^{\alpha,\beta}(\tau)\begin{cases}
\leq0,&\text{ for }\alpha\leq0 \text{ and }\beta<0\\\geq0,&\text{ for }\alpha\geq0 \text{ and } \beta>0\end{cases}$.
\item Let $\beta\geq0$ and $\alpha\geq 0$. Then 
$$A^{\alpha,\beta}(\tau)\in\big[-2\frac{\kappa\theta}{\chi^2}\ln\{1+T\kappa\},3
\frac{\kappa^2\theta T}{\chi^2}\big]$$
for all $\tau\in[0,T]$.
\item Let $a\not=0$. For all $c_2\in(\frac{\kappa-a}{\chi^2},\frac{\kappa+a}{\chi^2})$ it holds: if $\alpha< c_2$, then $B^{\alpha,\beta}(\tau)<c_2$.
\end{enumerate}
\label{LemmaPropCharF}
\end{lemma}
\begin{proof}$\mbox{}$\\
The proofs for Statements i), ii), iii) and iv) follow by trivial calculations. The remaining proofs are presented below.\\
\textbf{Statement v):} First we rewrite function $A^{\alpha,\beta}(\tau)$ in a convenient way by inserting the relation $1-c=\frac{2a}{-\alpha\chi^2+\kappa+a}$:
\begin{align*}
A^{\alpha,\beta}(\tau)&=\frac{\kappa\theta}{\chi^2}\Big[(\kappa-a)\tau-2\ln\Big\{\frac{1-c\exp(-a\tau)}{1-c}\Big\}\Big]\\
&=\frac{\kappa\theta}{\chi^2}\Big[\underbrace{(\kappa-a)}_{\in[0,\kappa]}\tau-2\ln\Big\{\underbrace{\frac{1-\exp(-a\tau)}{2a}}_{\in[0,\frac{T}{2}]}\underbrace{\big(-\alpha\chi^2+\kappa+a\big)}_{\in[0,2\kappa]}+\underbrace{\exp(-a\tau)}_{\in[\exp(-\kappa T),1]}\Big\}\Big].
\end{align*}
Now observe that $a\in[0,\kappa]$, as $\beta\geq0$. Thus, $\kappa-a\in[0,\kappa]$ and $\exp(-a\tau)\in[\exp(-\kappa\tau),1]$. Further, $-\alpha\chi^2+\kappa+a\in[0,2\kappa]$, as $\alpha\geq0$. Now consider the term $\frac{1-\exp(-a\tau)}{2a}$ and prove that it is monotonically decreasing by showing that its derivative w.r.t. $a$ is negative:
\begin{align*}
\frac{\partial}{\partial a}\Big(\frac{1-\exp(-a\tau)}{2a}\Big)=\frac{\exp(-a\tau)(1+a\tau)-1}{2a^2},
\end{align*}
where negativity follows by the general inequality $\exp(x)>1+x,\forall x\in\IR$. So, for $a\in[0,\kappa]$, $\frac{1-\exp(-a\tau)}{2a}\in[\frac{1-\exp(-\kappa\tau)}{2\kappa},\lim_{a\downarrow0}\frac{1-\exp(-a\tau)}{2a}]$. The limit is given by:
\begin{align*}
\lim_{a\downarrow0}\frac{1-\exp(-a\tau)}{2a}=\lim_{a\downarrow0}\frac{\exp(-a\tau)\tau}{2}=\frac{\tau}{2}.
\end{align*}
As $\tau\in[0,T]$, we obtain: $\frac{1-\exp(-a\tau)}{2a}\in[0,\frac{T}{2}]$. Combining the inequalities from above leads to the statement.\\
\textbf{Statement vi):} In this proof we consider $B^{\alpha,\beta}(\tau)$ as a function in $\alpha$ and fix all other parameters. Computing the first two derivatives shows that $B^{\alpha,\beta}$ is a convex, monotonically increasing function of $\alpha$:
\begin{align*}
\frac{\partial}{\partial\alpha}B^{\alpha,\beta}(\tau)&=\frac{4a^2\exp(-a\tau)}{(1-c\exp(-a\tau))^2(-\alpha\chi^2+\kappa+a)^2}\geq0\\
\frac{\partial^2}{\partial\alpha^2}B^{\alpha,\beta}(\tau)&=\frac{8a^2\chi^2\exp(-a\tau)(1-\exp(-a\tau))}{(\underbrace{1-c\exp(-a\tau)}_{>0})^3(\underbrace{-\alpha\chi^2+\kappa+a}_{>0})^3}\geq0.
\end{align*}
Further,
\begin{align*}
\lim_{\alpha\uparrow\frac{\kappa+a}{\chi^2}}B^{\alpha,\beta}(\tau)=\lim_{c\uparrow\infty}B^{\alpha,\beta}(\tau)=\frac{\kappa+a}{\chi^2}.
\end{align*}
Now we would like to find the points where the graph of $B^{\alpha,\beta}$ crosses the graph of function $f(\alpha)=\alpha$. To this aim we solve the following equation:
\begin{align*}
B^{\alpha,\beta}(\tau)&=\alpha\\
\Leftrightarrow (\kappa+a)(1-c\exp(-a\tau))-2a&=\chi^2\alpha(1-c\exp(-a\tau))\\
\Leftrightarrow (\alpha\chi^2-\kappa+a)(\exp(-a\tau)-1)&=0.
\end{align*}
Now assume that $\tau\not=0$ and $a\not=0$ and observe that the only solution is given by $\alpha=\frac{\kappa-a}{\chi^2}$. What is more, in this case the first two derivatives are even strictly positive, which means that $B^{\alpha,\beta}$ is strictly monotonically increasing and strictly convex in $\alpha$. Thus, its graph stays for $\alpha\in\big(\frac{\kappa-a}{\chi^2},\frac{\kappa+a}{\chi^2}\big)$ under the graph of the function $f(\alpha)=\alpha$, crosses it at $\alpha=\frac{\kappa-a}{\chi^2}$ and converges to it for  $\alpha\uparrow\frac{\kappa+a}{\chi^2}$. This proves Statement vi) for $\tau\not=0$ and $a\not=0$.\\
Now assume $a>0$ and $\tau=0$. Then $B^{\alpha,\beta}(0)=\alpha$ and Statement vi) follows directly in this case.  \\
$\mbox{}$\halmos
\end{proof}
\begin{notation}
With functions $\kappa$, $\theta$ and $\chi$ given as in Model (\ref{ModelHestTimeDep}) we introduce the following notation for all $e\in\E$:
\begin{align*}
{A}^{\alpha,\beta,e}(\tau)&:=-\frac{  \kappa(e)  \theta(e)(  \kappa(e)-  a(e))}{\chi(e)^2}\tau+\frac{2  \kappa(e)  \theta(e)}{\chi(e)^2}\ln\Big\{\frac{1-  c(e)\exp(-  a(e)\tau)}{1-  c(e)}\Big\}\\
  B^{\alpha,\beta,e}(\tau)&:=-\frac{-  c(e)(  \kappa(e)+  a(e))\exp(-  a(e)\tau)+  \kappa(e)-  a(e)}{\chi(e)^2\big(1-  c(e)\exp(-  a(e)\tau)\big)},
\end{align*}
where:
\begin{align*}
  a(e):=\sqrt{\kappa(e)^2+2\beta\chi(e)^2}, c(e):=\frac{\alpha\chi(e)^2+  \kappa(e)-  a(e)}{\alpha \chi(e)^2+  \kappa(e)+  a(e)}.
\end{align*}
Functions $\tilde a(e)$, $\tilde c(e)$, $\tilde{A}^{\alpha,\beta,e}(\tau)$ and $\tilde B^{\alpha,\beta,e}(\tau)$ are defined analogously for the transformed parameters $\tilde{\kappa}$ and $\tilde{\theta}$.
\end{notation}
Now we apply the previous two lemmas to derive the solution of the HJB equation. The result is presented in the theorem below.  
\begin{theorem}[Solution and verification in the time-dependent Heston model]
Assume the following conditions on the model parameters  in (\ref{ModelHestTimeDep}):
\begin{align}
\frac{1}{2\vartheta}\frac{\delta}{1-\delta}\frac{\big(\hat{\lambda}(e)\big)^2}{\big(\nu(e)\big)^2}&<\frac{\tilde\kappa^2(e)}{2\chi^2(e)},\forall e\in\E\label{AssumpBeta}\\
\max_{e\in\E}\Big\{\frac{\tilde\kappa(e)-\tilde a(e)}{\chi^2(e)}\Big\}&\leq \min_{e\in\E}\Big\{\frac{\tilde\kappa(e)+\tilde a(e)}{\chi^2(e)}\Big\}.\label{AssumpMaxMin}
\end{align}
Then the solution of the corresponding HJB system is given for all $(t,v,x)\in[0,T)\times \IR_{\geq 0}\times D_X$ by:
\begin{align*}
&\Phi^m(t,v,x)=\frac{v^\delta}{\delta}\IE\Big[\exp\Big\{\int_t^T\frac{1}{\vartheta} g(s,\tilde{X}^m(s),m(s))\text{d}s\Big\}\Big|\tilde{X}(t)=x\Big]^{\vartheta}\displaybreak[0]\\
&=\frac{v^\delta}{\delta}\Big[\exp\Big\{\int_t^T\frac{1}{\vartheta}\delta r(m(s))\text{d}s\Big\}\exp\big\{A_j(t_{j+1}-t)+B_j(t_{j+1}-t)x\big\}\prod_{i=j+1}^{k}\exp\{A_i(\tau_i)\}\Big]^{\vartheta}\displaybreak[0]\\
&=:\frac{v^\delta}{\delta}\exp\Big\{\int_t^T\delta r(m(s))\text{d}s\Big\}\exp\big\{\vartheta A^m(t)+\vartheta B^m(t)x\big\},
\end{align*}
where:
\begin{align*}
\tau_i&:=t_{i+1}-t_i,i=1,\beta_i:=\frac{1}{2\vartheta}\frac{\delta}{1-\delta}\frac{\big(\hat{\lambda}(m(t_i))\big)^2}{\big(\nu(m(t_i))\big)^2},\ldots,K\\
A_K&:=\tilde A^{0,\beta_K,m_K}(\tau_K), B_K:=\tilde B^{0,\beta_K,m_K}(\tau_K)\\
A_i&:=\tilde A^{B_{i+1},\beta_i,m_i}(\tau_i), B_i:=\tilde B^{B_{i+1},\beta_i,m_i}(\tau_i),i=0,\ldots,K-1.
\end{align*}
Further, $\Phi^m(t,v,x)$ is the value function of the optimization problem and the optimal portfolio strategy is given for all $t\in[0,T]$ by:
\begin{align*}
\bar{\pi}^m(t)=\frac{1}{1-\delta}\Big\{\frac{\hat{\lambda}(m(t))}{\big(\nu(m(t))\big)^2}+\rho\frac{\chi(m(t))}{\nu(m(t))}\vartheta B^m(t)\Big\}.
\end{align*}
\label{ThReprTimeDepHest}
\end{theorem}
\begin{proof}
The proof follows by recalling the probabilistic representation for $h^m$ from Equation (\ref{ProbReprTimeDeph}) and applying Lemma \ref{jointCharF} stepwise starting at the back. In each step the assumptions of Theorem \ref{jointCharF} need to be checked. They read as follows:
\begin{align}
\beta_i&=\frac{1}{2\vartheta}\frac{\delta}{1-\delta}\frac{\big(\hat{\lambda}(m(t_i))\big)^2}{\big(\nu(m(t_i))\big)^2}<\frac{\tilde\kappa^2(m_i)}{2\chi^2(m_i)},\forall i\in\{j,\ldots,K\}\label{betaCond}\\
\alpha_i&=\tilde B^{\alpha_{i+1},\beta_{i+1},m_{i+1}}(t_{i+2}-t_{i+1})<\frac{\tilde\kappa(m_i)+\tilde a(m_i)}{\chi^2(m_i)},\forall i\in\{j,\ldots,K-1\}\label{AlphaCond}\\
\alpha_K&=0<\frac{\tilde\kappa(m_K)+\tilde a(m_K)}{\chi^2(m_K)}.\label{alphaKCond}
\end{align}
Inequality (\ref{betaCond}) follows directly from Assumption (\ref{AssumpBeta}) and Inequality (\ref{alphaKCond}) is obvious as $\tilde{\kappa}(e),\tilde{a}(e)>0,\forall e\in\E$. For Inequality (\ref{AlphaCond}) recall that in our model $\beta_i>0$ and thus $\tilde a(m_i)<\tilde\kappa(m_i)$ for all $i\in\{0,\ldots,K\}$. Then, $\alpha_K=0<\max_{e\in\E}\Big\{\frac{\tilde\kappa(e)-\tilde a(e)}{\chi^2(e)}\Big\}:=c_2$. From Assumption (\ref{AssumpMaxMin}) we obtain further $c_2<\frac{\tilde\kappa(m_i)+\tilde a(m_i)}{\chi^2(m_i)}$ for all $i\in\{0,\ldots,K\}$. Statement (vii) from Lemma \ref{LemmaPropCharF} leads to $\tilde \alpha_{K-1}=\tilde B_K<c_2<\frac{\tilde\kappa(m_{K-1})+\tilde a(m_{K-1})}{\chi^2(m_{K-1})}$. To obtain Condition (\ref{AlphaCond}) for all $i$, observe that $\alpha_i=\tilde B_{i+1}$ and continue backwards in an analogous way showing that $\tilde B_{i+1}<c_2<\frac{\tilde\kappa(m_i)+\tilde a(m_i)}{\chi^2(m_i)}$ for all $i\in\{0,\cdots,K\}$.\\
Observe that here we do not need to check the assumptions of Corollary \ref{CorFK}, as we have derived the solution in an explicit form and it can be verified by direct substitution. So, $\Phi^m$ solves the corresponding HJB equation.\\
The verification result and the optimal portfolio strategy follow as a direct application of Theorem \ref{ThVerifTimeD}.\\
$\mbox{}$\halmos
\end{proof}

\subsection{Markov modulated Heston model with no leverage}
\label{SubsHestonNoRhoModel}
We continue with the Markov modulated Heston model with no correlation. After deriving the solution for the general parameter specifications (Section \ref{SecMMHNoRho}), we will simplify it in the separable case (Section \ref{SecSMMHNoRho}).
\subsubsection{General Markov modulated Heston model (MMH) with no leverage}
\label{SecMMHNoRho}
Consider Model (\ref{MCHeston}) and set $\rho=0$. Based on the results for the time-dependent model, the solution of the HJB equation in the Markov switching model can be derived as shown in the following theorem.
\begin{corollary}[Solution and verification in the MMH with $\rho=0$]$\mbox{}$\\
Let the conditions of Theorem \ref{ThReprTimeDepHest} hold. Then the value function $\Phi$ in Model (\ref{MCHeston}) with $\rho=0$ is given by the following equation:
\begin{align}
\Phi(t,v,x,e_i)=&\frac{v^{\delta}}{\delta}\IE\big[f^{\MC}(t,x)\big|\MC(t)=e_i\big]\notag\\
=&\frac{v^{\delta}}{\delta}\IE\big[\exp\Big\{\int_t^T\delta r(\MC(s))\text{d}s\Big\}\exp\big\{A^{\MC}(t)+B^{\MC}(t)x\big\}\big|\MC(t)=e_i\big],
\label{PhiHestonNonRhoCase1}
\end{align}
where for any $m\in\IM$, functions $A^{m}$ and $B^m$ are given by:
\begin{align*}
A^m(t)&=\sum_{j=0}^{K}\Big\{\int_t^T\delta r(m(s))\text{d}s+A_j(t_{j+1}-t)+\sum_{i=j+1}^{K}A_i(\tau_i)\Big\}1_{t\in[t_j,t_{j+1})}\\
B^m(t)&=\sum_{j=0}^{K}\big\{B_j(t_{j+1}-t)\big\}1_{t\in[t_j,t_{j+1})},
\end{align*} 
with:
\begin{align*}
\tau_i&:=t_{i+1}-t_i,\beta_i:=\frac{1}{2}\frac{\delta}{1-\delta}\frac{\big(\hat{\lambda}(m(t_i))\big)^2}{\big(\nu(m(t_i))\big)^2},i=1,\ldots,K\\
A_K&:=A^{0,\beta_K,m_K}(\tau_K), B_K:=B^{0,\beta_K,m_K}(\tau_K)\\
A_i&:=A^{B_{i+1},\beta_i,m_i}(\tau_i), B_i:=B^{B_{i+1},\beta_i,m_i}(\tau_i),i=0,\ldots,K-1.
\end{align*}
The optimal portfolio is:
\begin{align*}
\bar{\pi}(t)=\frac{1}{1-\delta}\frac{\hat{\lambda}(\MC(t))}{\big(\nu(\MC(t))\big)^2}.
\end{align*}
\label{CorHestonNoRhoGen}
\end{corollary}
\begin{proof} 
An application of Theorem \ref{ThReprTimeDepHest} for $\rho=0$, i.e. $\vartheta=1$, and Theorem \ref{ThVerifNoRho} leads to the statement. \\
$\mbox{}$\halmos
\end{proof}
Observe that function $\Phi$ can be easily computed by a partial Monte Carlo method, where one has to simulate only the path of the Markov chain and not all other processes.\\
\subsubsection{Separable Markov modulated Heston model (SMMH) with no leverage}
\label{SecSMMHNoRho}
Now we will consider the separable example corresponding to the special case presented in Section \ref{SubsunNoRhoSepa}. So, we specify our model in such a way that a separable explicit solution can be found:
\begin{align}
\begin{aligned}
&\text{d}  P_0(t) = P_0(t)r{\big(\MC(t)\big)}\text{d}t\\
&	\text{d} P_1(t)=P_1(t)\Big[r\big(\MC(t)\big)+\underbrace{d\nu\big(\MC(t)\big)}_{=\hat{\lambda}(\MC(t))}X(t)\text{d} t+\nu\big(\MC(t)\big)\sqrt{X(t)}\text{d} W_P(t)\Big]\\
&	\text{d}{X}(t)=\kappa(\theta\big(\MC(t)\big)-X(t))\text{d}t+\chi\sqrt{X(t)}\text{d}W_X(t)\\
&	\text{d}\langle W_P, W_X\rangle(t)=0,
\label{MCHestonNonRhoCase3}
\end{aligned}
\end{align}
where $P_0(0)=p_0$, $P_1(0)=p_1$, $X(0)=x_0$, $\kappa,\chi\in\IR_{>0}$, $d\in\IR$, and as before $\theta(e)\in\IR_{>0}$ and $2\kappa\theta(e)\geq\chi^2$, for all $e\in\E$, so that $X(t)\geq0$, for all $t\in[0,T]$. This model can be embedded in the notation from (\ref{SMMAF}) as follows:
\begin{align}
\begin{aligned}
\gamma^2(x,e_i)&=d^2x&&\Rightarrow\Gamma^{(1)}(e_i)=0,\ \ \bar\Gamma^{(2)}=d^2\\
\mu_X(x,e_i)&=\kappa\theta(e_i)-\kappa x&&\Rightarrow\mu_X^{(1)}(e_i)=\kappa\theta(e_i),\ \ \bar\mu_X^{(2)}=-\kappa\\
\sigma^2_X(x,e_i)&=\chi^2x&&\Rightarrow\Sigma_X^{(1)}(e_i)=0,\ \ \bar\Sigma_X^{(2)}=\chi^2.
\end{aligned}
\tag{SMMH}
\label{SMMH}
\end{align}
A direct application of Theorem \ref{ThVerifTimeD} and Theorem \ref{ThExplSolNoRhoSpecCase} leads to the following solution:
\begin{corollary}[Solution and verification in the SMMH with $\rho=0$]$\mbox{}$\\
Consider Model (\ref{MCHestonNonRhoCase3}) and assume:
\begin{align}
\frac{\delta}{1-\delta}d^2<\frac{\kappa^2}{\chi^2}.
\label{AssumpBetaCase2}
\end{align}
Then the value function is given for all $(t,v,x,e_i)\in[0,T]\times\IR_{\geq0}\times D_X\times\E$ by:
\begin{align}
\Phi(t,v,x,e_i)=\frac{v^{\delta}}{\delta}\bar\xi(t,e_i)\exp\{B(t)x\},
\label{PhiHestonNonRhoCase3}
\end{align}
where:
\begin{align*}
\bar\xi(t,e_i)=\IE\Big[\exp\Big\{\int_t^Tw(s,\MC(s))\text{d}s\Big\}\Big|\MC(t)=e_i\Big], \forall\ (t,e_i)\in [0,T]\times\E,
\end{align*}
with $w(t,e_i)=\delta r(e_i)+B(t)\kappa\theta(e_i)$, for all $(t,e_i)\in [0,T]\times\E$, and 
\begin{align*}
B(t)=\frac{-c(\kappa+a)\exp\{-a(T-t)\}+\kappa-a}{\chi^2\big(1-c\exp\{-a(T-t)\}\big)}, \forall\ t\in [0,T],
\end{align*}
where $a=\sqrt{\kappa^2-\frac{\delta}{1-\delta}d^2\chi^2}$ and $c:=\frac{\kappa-a}{\kappa+a}$. The optimal portfolio is:
\begin{align*}
\bar{\pi}(t)=\frac{1}{1-\delta}\frac{d}{\nu(\MC(t))}.
\end{align*}
\end{corollary}
\begin{proof}
Follows directly by Theorem \ref{ThExplSolNoRhoSpecCase} and \ref{ThVerifTimeD}.\\
$\mbox{}$\halmos
\end{proof}
\subsection{Separable Markov modulated Heston model with leverage ($\text{SMMH}^{\rho}$)}
\label{SubsHestonRhoModel}
For a tractable model for the case with leverage we generalize Model (\ref{MCHestonNonRhoCase3}) to $\rho\not=0$:
\begin{align}
\begin{aligned}
&\text{d}  P_0(t)= P_0(t)r{\big(\MC(t)\big)}\text{d}t\\
&	\text{d} P_1(t)=P_1(t)\Big[r\big(\MC(t)\big)+ d\nu\big(\MC(t)\big)X(t)\text{d} t+\nu\big(\MC(t)\big)\sqrt{X(t)}\text{d} W_P(t)\Big]\\
&	\text{d}{X}(t)=\kappa\big\{\theta\big(\MC(t)\big)-X(t)\big\}\text{d}t+\chi\sqrt{X(t)}\text{d}W_X(t)\\
&	\text{d}\langle W_P, W_X\rangle(t)=\rho\text{ d}t,
	\label{MCHestonRhoCase1}
\end{aligned}
\end{align}
with initial values $P_0(0)=p_0$, $P_1(0)=p_1$, $X(0)=x_0$. This corresponds to the following specifications in the notation from (\ref{SMMAFrho}):
\begin{align}
\begin{aligned}
\gamma^2(x,e_i)&=d^2x&&\Rightarrow\Gamma^{(1)}(e_i)=0,\ \ \bar\Gamma^{(2)}=d^2\\
\mu_X(x,e_i)&=\kappa\theta(e_i)-\kappa x&&\Rightarrow\mu_X^{(1)}(e_i)=\kappa\theta(e_i),\ \ \bar\mu_X^{(2)}=-\kappa\\
\sigma^2_X(x,e_i)&=\chi^2x&&\Rightarrow b=\frac{\chi}{|d|}.
\end{aligned}
\tag{$\text{SMMH}^{\rho}$}
\label{SMMHrho}
\end{align}
As before Theorem \ref{ThVerifTimeD} and Theorem \ref{ThVerifRho} lead to the verification result in this case:
\begin{corollary}[Solution and verification in $\text{SMMH}^{\rho}$]$\mbox{}$\\
\label{CorMSHestonRhoSol}
Assume that:
\begin{align}
0&<\kappa-\frac{\delta}{1-\delta}\rho \chi |d|\label{cond3}\\
\frac{\delta}{1-\delta}  d^2&<\frac{\vartheta (\kappa-\frac{\delta}{1-\delta}\rho \chi |d| )^2}{\chi^2}\label{cond1}
\end{align}
for $a:=\sqrt{(\kappa-\frac{\delta}{1-\delta}\rho \chi |d| )^2-\frac{\delta}{1-\delta}\frac{\chi^2}{\vartheta}  d^2}$. Then, the value function in Model (\ref{MCHestonRhoCase1}) is given for all $(t,v,x,e_i)\in[0,T]\times\IR_{\geq0}\times D_X\times\E$ by: 
\begin{align}
\Phi(t,v,x,e_i)=\frac{v^{\delta}}{\delta}\xi(t,e_i)\exp\{D(t)x\},
\end{align}
where:
\begin{align}
\xi(t,e_i)=\IE\Big[\exp\Big\{\int_t^T\upsilon(s,\MC(s))\Big\}\text{d}s\Big|\MC(t)=e_i\Big], \forall (t,e_i)\in [0,T]\times\E,
\label{xiHestonCor}
\end{align}
for $\upsilon(t,e_i)=\delta r(e_i)+D(t)\kappa\theta(e_i)$, and:
\begin{align}
\begin{aligned}
D(t)=\frac{\vartheta\big(-c(\kappa-\frac{\delta}{1-\delta}\rho \chi   |d| +a)\exp\{-a(T-t)\}+\kappa-\frac{\delta}{1-\delta}\rho \chi |d| -a\big)}{\chi^2(1-c\exp\{-a(T-t)\})},\forall t\in [0,T]
\end{aligned}
\label{solutionDrho}
\end{align}
with $c:=\frac{\kappa-\frac{\delta}{1-\delta}\rho \chi |d| -a}{\kappa-\frac{\delta}{1-\delta}\rho \chi |d| +a}$.
The optimal portfolio is:
\begin{align}
\bar{\pi}(t)=\frac{1}{1-\delta}\Big[\frac{d }{\nu\big(\MC(t)\big)}+\rho\frac{ \chi}{\nu\big(\MC(t)\big)}D(t)\Big].
\label{HestonRhoOptimPortf}
\end{align}
As mentioned earlier the first part of the optimal portfolio, $\bar\pi_{MV}(t):=\frac{1}{1-\delta}\frac{d}{\nu(\MC(t))}$, is called the mean-variance term and the second one, $\bar\pi_H(t):=\frac{1}{1-\delta}\rho\frac{\chi}{\nu(\MC(t))}D(t)$, is the hedging term.
\end{corollary}
\begin{proof}
The statement follows from Theorem \ref{ThVerifRho} and Theorem \ref{ThVerifTimeD}.\\
$\mbox{}$\halmos
\end{proof}
After having verified theoretically this solution we are now interested in interpreting it from an economical point of view. We start by deriving explicitly the dynamics of the variance process $Var(t):=\nu^2(\MC(t))X(t)$ of the log returns of the risky asset:
\begin{corollary}[Variance of the log returns] $\mbox{}$\\
The instantaneous variance of the log asset returns is characterized by the following SDE:
\begin{align*}
\text{d}Var(t)&=\hat{\kappa}\big(\MC(t)\big)\big(\hat{\theta}\big(\MC(t)\big)-Var(t)\big)\text{d}t+\hat{\chi}\big(\MC(t)\big)\sqrt{Var(t)}\text{d}W_X(t)\\
&+Var(t)\sum_{i=1}^{l}\frac{\nu^2(e_i)}{\nu^2\big(\MC(t)\big)}\text{d}M_i(t),
\end{align*}
where:
\begin{align*}
\hat{\kappa}\big(\MC(t)\big)&=\kappa-\sum_{i=1}^{l}\frac{\nu^2(e_i)}{\nu^2\big(\MC(t)\big)}q_{MC(t)i}\\ 
\hat{\theta}\big(\MC(t)\big)&=\frac{\nu^2\big(\MC(t)\big)\kappa\theta\big(\MC(t)\big)}{\hat{\kappa}}\\
\hat{\chi}\big(\MC(t)\big)&=\nu\big(\MC(t)\big)\chi.
\end{align*}
With this notation the price process for the risky asset is given by the following SDE:
\begin{align*}
\text{d}P_1(t)=P_1(t)\Big[r\big(\MC(t)\big)+\frac{d}{\nu\big(\MC(t)\big)}Var(t)\Big]\text{d}t+\sqrt{Var(t)}\text{d}W_P(t).
\end{align*}
\end{corollary}
\begin{proof}
Follows directly as an application of It\^o's formula for Markov-modulated It\^o diffusions, which is stated in Theorem \ref{ThIto}.
\end{proof}
Observe that the variance $Var$ follows a mean reverting process with jumps according to the Markov chain, where all parameters depend on the Markov chain. This rich stochastic structure makes the considered model very flexible and suitable for describing a wide range of markets. Furthermore, note that the market price of risk defined as the excess return divided by $Var(t)$ is given by $\frac{d}{\nu(\MC(t))}$, which is exactly the driver of the mean variance term of the optimal portfolio. The higher the reward for the corresponding risk (i.e. the higher $d$) or the lower the risk for the same reward (i.e. the lower $\nu$), the more attractive the stock becomes for the investor and the bigger $\bar\pi_{MV}$. Note that the absolute value of $\bar\pi_H$ exhibits similar behavior when $\nu$ changes, as the quotient between $\bar\pi_{MV}$ and $\bar\pi_H$ is invariant w.r.t. $\nu$:
$$\frac{\bar\pi_{MV}(t)}{\bar\pi_H(t)}=\frac{d}{\rho\chi D(t)}.$$
What is more, this relation remains the same for all states of the Markov chain. This makes sense as the hedging term is a correction of the mean variance portfolio, so it is adjusted whenever $\bar\pi_{MV}$ changes because of the Markov chain.\\
One can understand the influence of $\nu$ on the optimal portfolio even better by recalling the SDE of the wealth process resulting from the optimal strategy given in Equation (\ref{HestonRhoOptimPortf}):
\begin{align*}
\text{d}V^{\bar\pi}(t)=V^{\bar\pi}(t)\Big\{&r\big(\MC(t)\big)+\frac{d}{1-\delta}\big(d+\rho\chi D(t)\big)X(t)\text{d}t\\
&+\frac{1}{1-\delta}\big(d+\rho\chi D(t)\big)\sqrt{X(t)}\text{d}W_P(t)\Big\}.
\end{align*}
Note that the wealth process does not depend on $\nu$ and recall that $\nu$ determines the market price of risk and the influence of the stochastic factor $X$ on the volatility of the asset log returns. So, the optimal portfolio is chosen in such a way, that the investor is protected against changes in $\nu$ and the influence of this parameter. Another interesting observation is that the mean reversion level $\theta$ and the current level of $X$ do not influence directly the optimal policy. The reason is that the stochastic factor $X$ influences proportionally the instantaneous variance and the excess return of the risky asset price.\\
A detailed analysis of the single components driving the optimal portfolio is given in the next section. 
\subsection{Numerical implementation and discussion of the results}
\label{SubHestonNumRho}
In this section we illustrate and interpret the results derived above by some numerical examples. First of all we will specify the basic parameter set we are working with and show the numerical results for this parameter specification in order to point out the impact of the Markov switching. Then we will analyze the influence of the risk aversion parameter $\delta$ on the behavior of the investor. We will continue with the impact of $d$ and $\nu$, which are the driving parameters of $\bar\pi_{MV}$ as one can see in Equation (\ref{HestonRhoOptimPortf}). Afterwards we will discuss the sensitivity of the results to changes of the remaining parameters.\\
In what follows we consider Model (\ref{MCHestonRhoCase1}) and assume that the Markov chain can switch between two states. The first one, $e_1$, describes a calm market with moderate volatility levels. The second one, $e_2$ corresponds to a turbulent state with higher volatility and lower market price of risk. The investment time horizon is set to $T=5$ throughout the whole section. Based on the empirical results from \cite{AitSahalia2007}\footnote{In this paper the parameters of a Heston model are estimated based on daily  observations of the stock index S\&P500 and the volatility index VIX over the period from 1990 till 2003. Based on their results we choose our Markov modulated parameters in such a way, that the first state of $\MC$ describes a calm market and the second a volatile one.} we fix the following basic parameter set: $\kappa=4$, $\theta(e_1):=\theta_1=0.02$, $\theta(e_2):=\theta_2=0.04$, $\chi=0.35$, $d=1.7$, $\nu(e_1):=\nu_1=1$, $\nu(e_2):=\nu_2=1.3$, $r(e_1):=r_1=0.03$, $r(e_2):=r_2=0.01$, $\rho=-0.8$. So, on average the time needed for the mean reversion of process $X$ is $\frac{1}{\kappa}=\frac{1}{4}=3$ months. Further, following \cite{Bernhart2011}\footnote{In this paper a Markov-modulated Black Scholes model is estimated to weekly prices of S\&P500 over the period from 1987 till 2009. We use their result for the intensity matrix of the Markov chain.}, we set the elements of the intensity matrix to $q_{11}=-1.0909$, $q_{22}=-3.4413$. This means that on average the Markov chain remains one year in the calm state and around 4 months in the turbulent one, as the waiting time the Markov chain spends in state $e_i$ before the next jump is exponentially distributed with parameter $-q_{ii}$ and expectation $\frac{1}{q_{ii}}$. We will compare the results for two investors with different risk preferences: the first one has a positive risk aversion parameter $\delta=0.3$ and for the second one it is negative $\delta=-1$. We denote these parameter specifications by Set 1 and Set 2, respectively. This differentiation is necessary because $\delta$ influences strongly the optimal behavior of the investor, as we will see in what follows.\\
Table \ref{tab:CompFormSim} contains the optimal expected utility for Set 1 and Set 2 calculated as in Corollary \ref{CorMSHestonRhoSol}, where function $\xi$ is computed by a Monte Carlo simulation of the Markov chain with $10$ thousand simulations. Additionally, it is compared to the optimal expected utility computed via a full Monte Carlo simulation of Model (\ref{MCHestonRhoCase1}) with $1$ Mio. simulations and 250 steps per year, i.e. trading is done once per day.
\begin{table}[H]
\centering
\begin{tabular}{|c||c|c|c|c|}
\hline
Parameter & Formula & Comp. time & Monte Carlo & Comp. time\\
\hline
Set 1 & 7.4261 & 40 sec & 7.4260 & approx. 2.2 h \\
Set 2 & -0.0802 & 40 sec & -0.0802 & approx. 2.2 h \\
\hline
\end{tabular}
\caption{Comparison of the optimal expected utility computed as in Corollary \ref{CorMSHestonRhoSol} (second column) and by a full Monte Carlo simulation(fourth column), where $V(0)=10$, $X(0)=0.02$, $\MC(0)=e_1$, $T=5$. The third and fifth columns contain the corresponding computational time.}
\label{tab:CompFormSim}
\end{table}
One can notice that the values from the second and the fourth columns are quite close to each other. This shows that trading only once per day, which might be sensible in reality, does not lead to significant loss of utility. So the derived optimal strategy is practically applicable. What is more, note that many time consuming Monte Carlo simulations are required in order to obtain converging results for the full Monte Carlo approach, which confirms the importance of the derived theoretical results in Corollary \ref{CorMSHestonRhoSol}.\\
Now we compare the optimal trading strategies in the two discussed parameter settings. They are presented in Figure \ref{fig:PIset12}. The main part of the optimal strategy is given by the mean-variance portfolio. It can be observed that it is positive, which is to be expected, as the expected asset return exceeds the riskless interest rate for both states of the Markov chain. One can also recognize that higher $\delta$ leads to a higher long position in the risky asset and for $\delta=0.3$ the exposure even exceeds the investors's wealth. Lower $\delta$ results in a more moderate investment in the risky asset. This observation is in accordance to the interpretation of $\delta$ as a risk aversion parameter. Furthermore, the investment in the risky asset is lower in the turbulent state of the Markov chain, as it is associated with lower excess return and higher volatility, thus higher risk. It is because of this difference in the optimal behavior for the different states of the Markov chain, that it is important for the investor to recognize the true state and to react to the parameter changes.
\begin{figure}[h]
    \centering
    \includegraphics[trim=0.3cm 8cm 10.5cm 8cm,clip,scale=0.5]{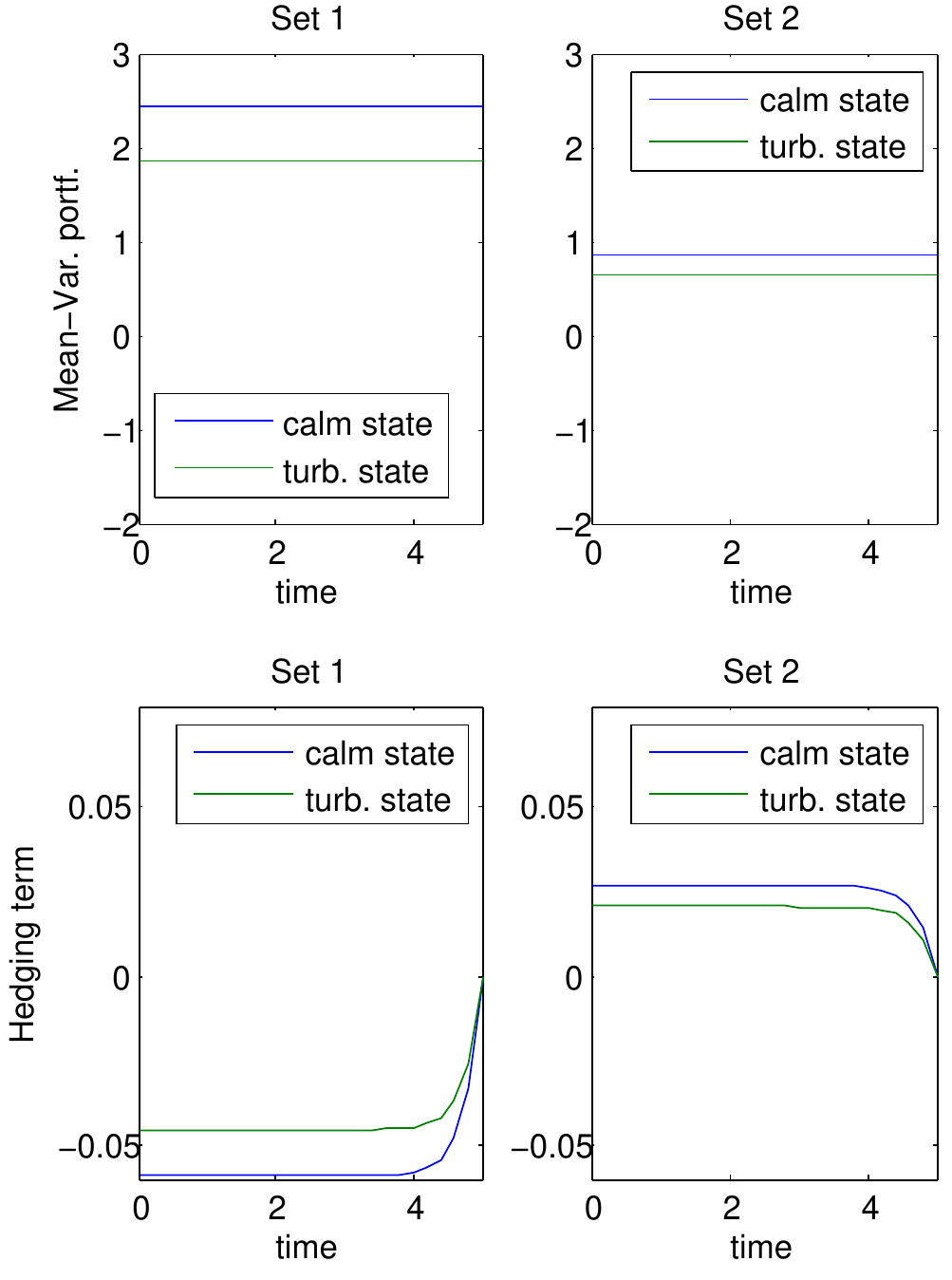}
    \caption{Mean-variance portfolio (first row) and hedging term (second row) for the two considered parameter sets. The blue lines represent the optimal investment in the calm state and the green lines the optimal investment in the turbulent state. }
    \label{fig:PIset12}
\end{figure}
The hedging term depends among others on the correlation $\rho$ between the Brownian motions for the stock and the stochastic factor, the volatility coefficient $\chi$ of the stochastic factor and function $D$. Function $D$ is plotted in Figure \ref{fig:Y1set1234}. 
\begin{figure}[h]
    \centering
    \includegraphics[trim=0cm 15cm 0cm 8cm,clip,scale=0.45]{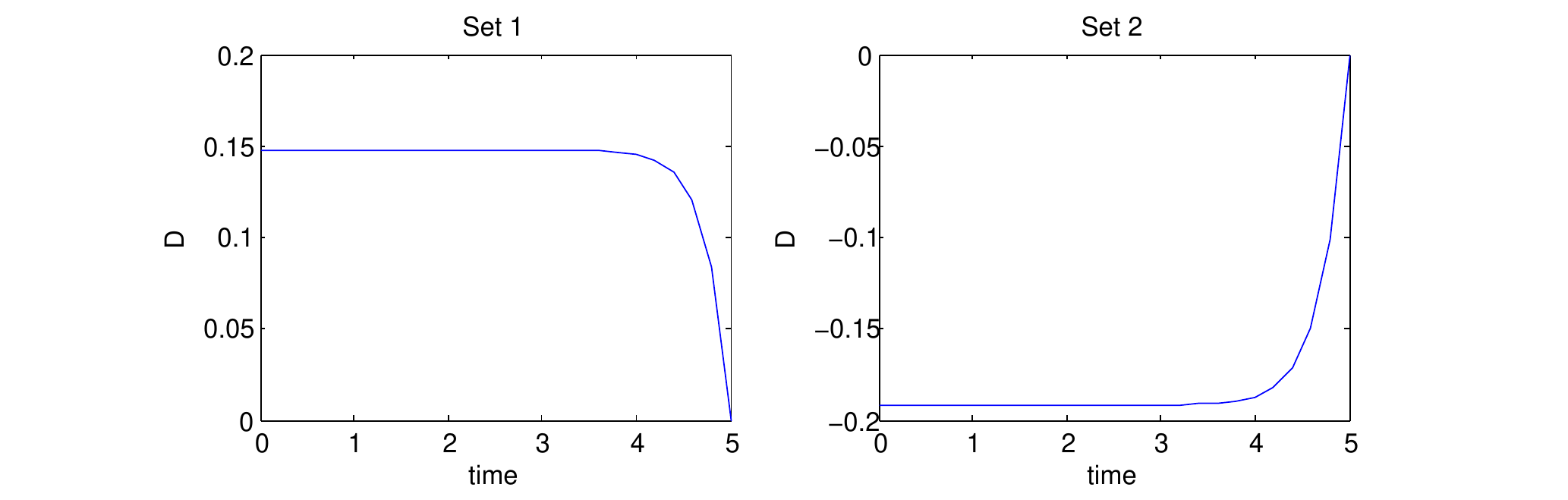}
    \caption{Function $D$ against time for the two different parameter sets.}
    \label{fig:Y1set1234}
\end{figure}
Note that for $\delta>0$ it is positive, a fact that follows from Statements i), ii) and iii) in Lemma \ref{LemmaPropCharF}. So, as $\rho$ is negative, the hedging term for Set 1 is negative. This makes sense, as the negative correlation relates higher volatility due to higher increments of $W_X$ to falling asset prices. So, the investor reduces his long risky position as a protection against this additional risk. The hedging term in the turbulent state is slightly higher, i.e. it has a smaller absolute value, as the mean-variance portfolio is also smaller in this state, the quotient $\frac{\bar\pi_{H}}{\bar\pi_{MV}}$, however, remains the same, as mentioned before. If $\delta<0$ then function $D$ is negative and because of the negative $\rho$ the hedging term is positive in both states. This might be interpreted as a hedge against missed profit chances, as the mean-variance portfolio of this highly risk averse investor is smaller than one. Note that for both values for $\delta$ the absolute value of $D$ is remaining stable over the almost whole time horizon and in the last year it decreases rapidly towards zero. The role of the hedging diminishes for short time horizons, as small changes in the price of the risky asset influence strongly the final wealth and this high sensitivity, thus high risk, is dominating the hedging effect.\\
Now we deepen our observations on the influence of the risk aversion parameter $\delta$ on the optimal portfolio. Figure \ref{fig:Histograms} shows the density of the terminal wealth of an investor following the derived optimal strategy for different values of $\delta$.
\begin{figure}[h]
    \centering
    \includegraphics[trim=0cm 4cm 0cm 3.8cm,clip,scale=0.5]{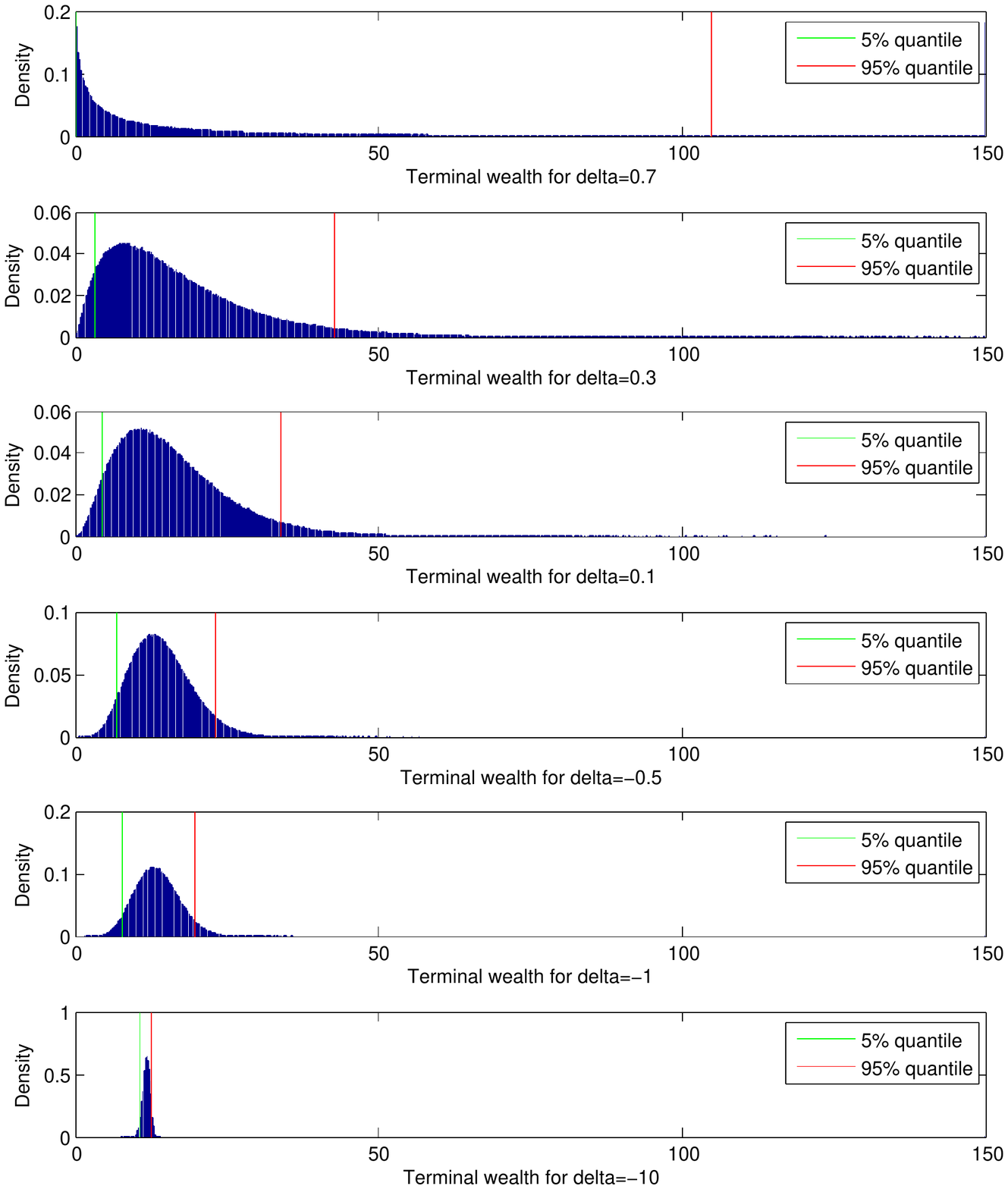}
    \caption{Density of the terminal wealth of an investor following the derived optimal strategy for different values of $\delta$, where the remaining parameters are adopted from Set 1. The densities are obtained via Monte Carlo simulations of Model (\ref{MCHestonRhoCase1}). For reasons of better comparability all values higher than 150 are summarized in the last bar in the plots.}
    \label{fig:Histograms}
\end{figure}
One can clearly recognize that higher values for $\delta$ lead to higher probabilities for both very low (close to zero) and very high (even above 120) wealth levels. In contrast, for small $\delta$ both probabilities for high losses and high profits are much smaller. This fact is reflected also in the shift of the $5\%$-quantile to the right and of the $95\%$-quantile to the left for smaller values for $\delta$.
Figure \ref{fig:PIdelta} confirms that the smaller $\delta$, the more conservative the mean-variance portfolio. The absolute value of the hedging term also decreases for smaller delta mainly because of the influence of the factor $\frac{1}{1-\delta}$ which drives the mean-variance term as well. Furthermore, in Figure \ref{fig:PIdelta} one can observe once again that in the turbulent state the investor holds less of the risky asset throughout time and all values for $\delta$.
\begin{figure}[h]
    \centering
    \includegraphics[trim=0cm 7cm 0cm 7cm,clip,scale=0.5]{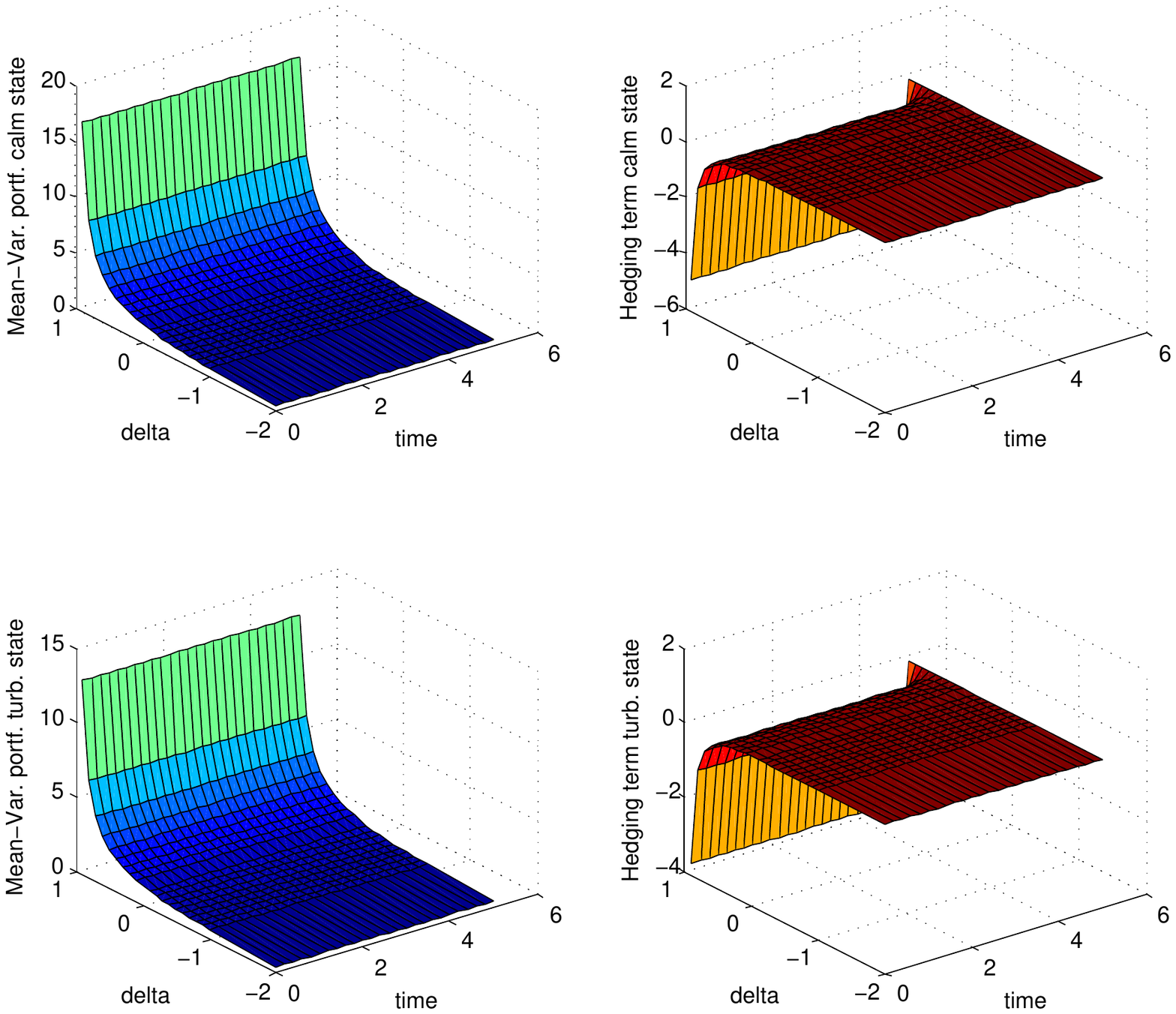}
    \caption{Optimal mean-variance portfolio and hedging term over time for different values of $\delta$. The remaining parameters are adopted from Set 1.}
    \label{fig:PIdelta}
\end{figure}
%
We proceed with the remaining components of the mean-variance portfolio: $d$ and $\nu$. The influence of $d$ on $\bar\pi_{MV}$ is naturally positive, if $\nu$ is positive, as high $d$ means high market price of risk. The absolute value of the hedging term also increases with $d$ which might be interpreted in the case of $\delta>0$ as a compensation for the higher exposure and in the case of $\delta<0$ as willingness to increase the exposure due to the higher market price of risk. Figure \ref{fig:Histogramd} shows the influence of $d$ on the distribution of the terminal wealth of an investor with $\delta=0.3$. One can see that the higher $d$, the higher the $95\%$ quantile on the one side and the lower the $5\%$ quantile on the other. So, both probabilities for big gains and big losses get higher. The risk on the down side comes from the fact that the ivestor has a higher exposure for bigger values for $d$, so if the stock price falls, she bears higher losses. However, the effect on the positive side is stronger, as an increasing $d$ leads to a high excess return of the stock and thus to the possibility for much higher gains. As Figure \ref{fig:Histogramd1} shows, the same can be observed also for $\delta=-1$, however the influence of $d$ on the wealth distribution is not so strong, as we are dealing with a more risk averse investor, who prefers having less exposure to the risky asset. Furthermore, observe that the wealth distribution in this case is more symmetric than for $\delta=0.3$, which reflects once again the risk aversion of the investor. In contrast, an investor with $\delta=0.3$ is willing to accept the high mass on the lower end because it is compensated by the possibility for very high profits.\\
\begin{figure}[h]
    \centering
    \includegraphics[trim=0cm 11.5cm 0cm 3.8cm,clip,scale=0.5]{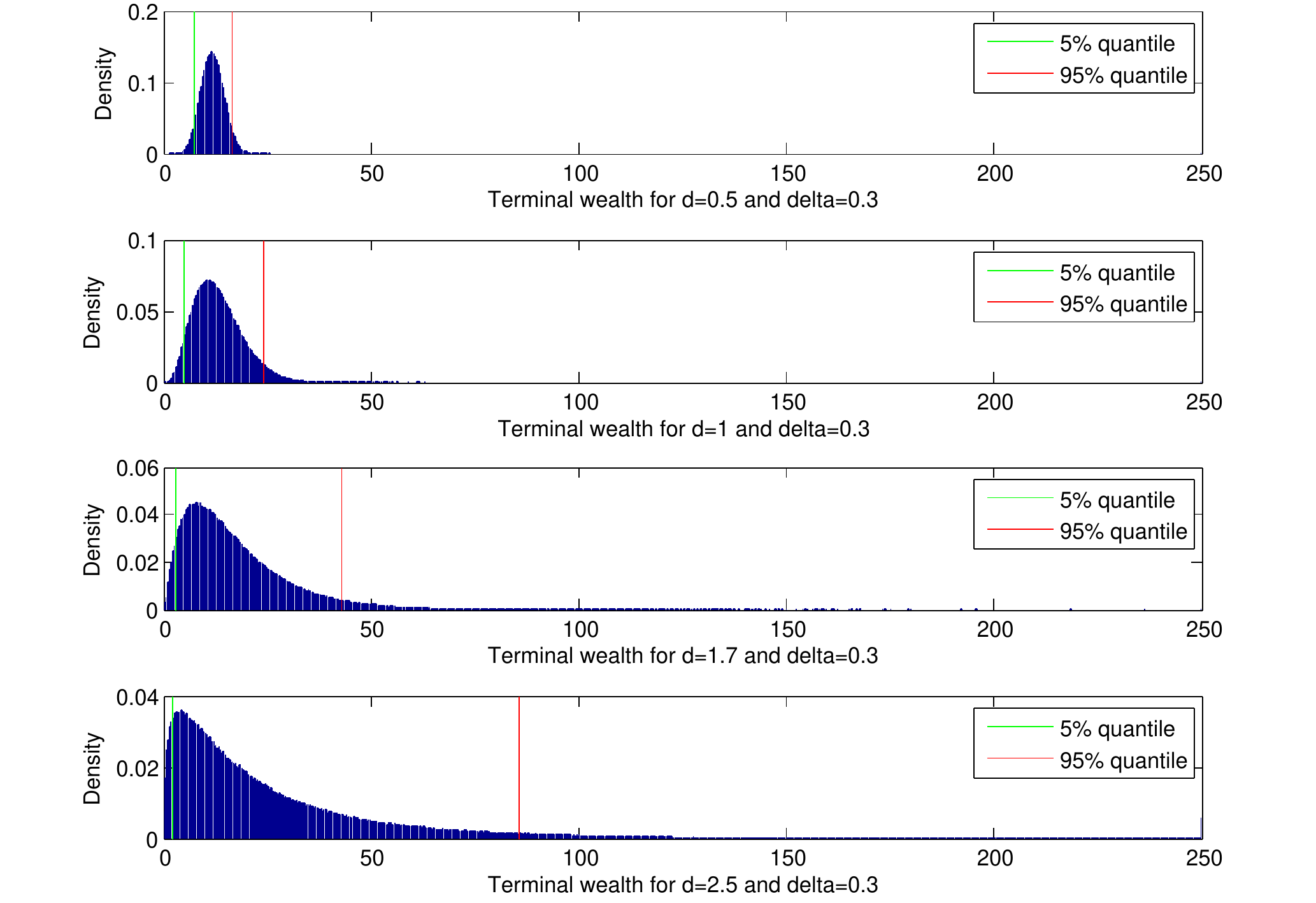}
    \caption{Density of the terminal wealth of an investor following the derived optimal strategy for different values of $d$, where the remaining parameters are adopted from Set 1. The densities are obtained via Monte Carlo simulations of Model (\ref{MCHestonRhoCase1}). For reasons of better comparability all values higher than 250 are summarized in the last bar in the plots.}
    \label{fig:Histogramd}
\end{figure}
\begin{figure}[h]
    \centering
    \includegraphics[trim=0cm 11.5cm 0cm 3.8cm,clip,scale=0.5]{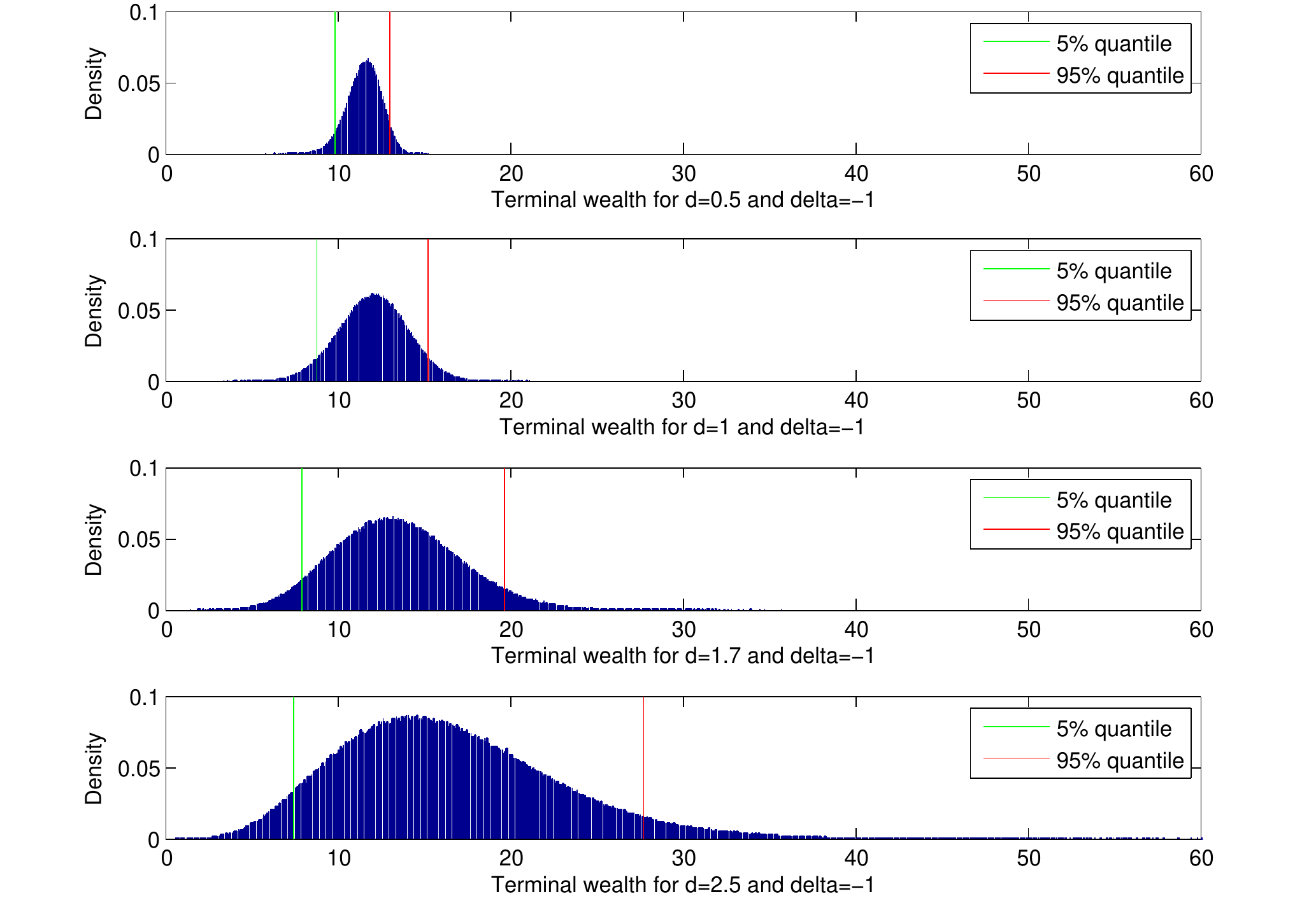}
    \caption{Density of the terminal wealth of an investor following the derived optimal strategy for different values of $d$, where the remaining parameters are adopted from Set 2. The densities are obtained via Monte Carlo simulations of Model (\ref{MCHestonRhoCase1}). For reasons of better comparability all values higher than 60 are summarized in the last bar in the plots.}
    \label{fig:Histogramd1}
\end{figure}
The influence of $\nu$ on the investment strategy can be easily derived from the analytical formula for the optimal portfolio. Higher values of $\nu$ lead to lower investments in the risky assets, as it reduces the market price of risk. As $\nu_1<\nu_2$, the exposure in the risky asset is smaller in the turbulent state than in the calm one. This observation holds for the mean-variance portfolio as well as for the hedging term. The bigger the difference between $\nu_1$ and $\nu_2$ the more important for the investor to recognize the Markov switching character of the market and to adjust her strategy.\\
Let us continue with the remaining model parameters. As we can easily see from Equation (\ref{HestonRhoOptimPortf}) $\rho$ influences the sign of the hedging term: $sign(\bar\pi_H)=sign(D)sign(\rho)=sign(\delta)sign(\rho)$. So for positive $\delta$ and positive $\rho$ the hedging term is positive because the investor profits from the possibility to achieve higher terminal wealth levels, as an increase of $X$ due to positive increments of $W_X$ would be related to an increase in the stock price due to positive increments in $W_P$). On the contrary for a very risk averse investor with $\delta<0$ positive correlation results in a negative hedging portfolio, as the possibility for high returns does not compensate enough for the increase in the risk of falling $X$ and falling stock prices. For $\delta>0$ and $\rho<0$ the hedging term is negative in order to reduce the exposure to the additional risk coming from an increase in stochastic factor and fall in the stock price, especially considering the big long position from $\bar\pi_{MV}$. If $\delta<0$ and $\rho<0$ the hedging term is positive and allows the investor to profit from scenarios with low volatility and thus stable high asset prices. This makes sense especially considering the interpretation of the hedging term as protection against too conservative investments. Now that we have interpreted the sign of the hedging term we will consider how its absolute value changes with the considered parameters. The higher the correlation between $W_X$ and $W_P$, the higher the absolute value of the hedging term as the stochastic factor can be better hedged by the risky asset. Furthermore, it can be summarized that lower values for $\chi$ and higher values for $\kappa$ lead to a lower absolute value of the hedging term, as lower volatility coefficient reduce the risk coming from the stochastic factor $X$ and faster mean reversion makes it more difficult to hedge. The corresponding plots are omitted here for space reasons.\\
Now we will discuss the influence of the Markov switching parameter $\theta$. It does not influence the optimal portfolio policy, as mentioned earlier, but only function $\xi$ and herewith the value function $\Phi$. Figures \ref{fig:Y2theta} and \ref{fig:Y2theta1} present function $\Phi$ for different choices for $\theta$ for $\delta=0.3$ and $\delta=-1$, respectively.%
\begin{figure}[h]
    \centering
    \includegraphics[trim=0cm 2.5cm 0cm 10.8cm,clip,scale=0.5]{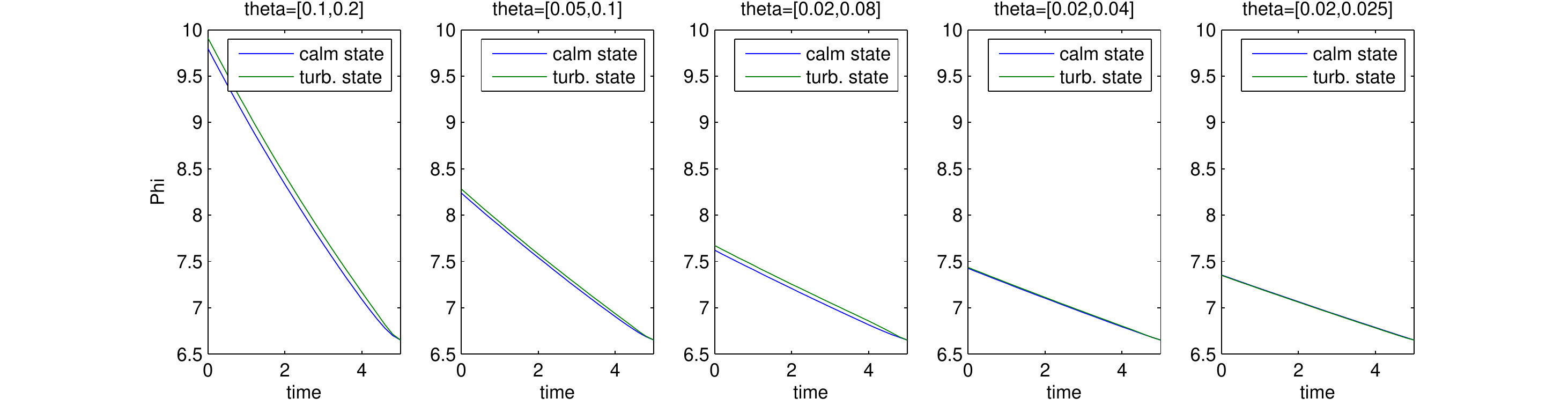}
    \caption{Function $\Phi$ against time for different choices for $(\theta_1,\theta_2)$. The remaining parameters are adopted from Set 1 ($\delta=0.3$).}
    \label{fig:Y2theta}
\end{figure}
\begin{figure}[h]
    \centering
    \includegraphics[trim=0cm 2.5cm 0cm 10.8cm,clip,scale=0.5]{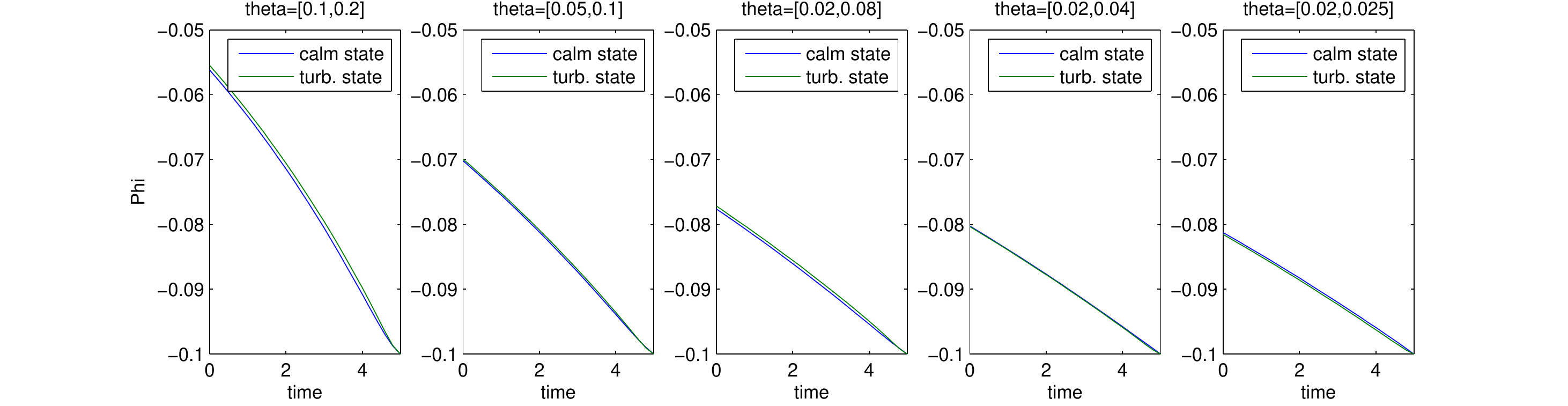}
    \caption{Function $\Phi$ against time for different choices for $(\theta_1,\theta_2)$. The remaining parameters are adopted from Set 2 ($\delta=-1$).}
    \label{fig:Y2theta1}
\end{figure}
It can be seen that an increase in $\theta_1$ and $\theta_2$ leads to higher optimal expected utility both for $\delta=0.3$ and $\delta=-1$. The investor takes advantage of the stochastic factor and uses this additional risk to generate more profit. Furthermore, it can be seen from the plots that the value function $\Phi$ does not depend heavily on the initial state of the Markov chain. This might be explained by the fact that for longer maturities the Markov chain will have enough time to switch a few times and for shorter maturities the integral in the formula for $\xi$ is quite small. However, the bigger the difference between $\theta_1$ and $\theta_2$ the bigger the influence of the current state of the Markov chain on $\Phi$.\\
\section{Conclusion}
\label{SecConcl}
In the context of utility maximization we presented a flexible and at the same time analytically tractable asset price model, where two additional sources of randomness are considered: a Markov chain and a continuous stochastic factor following a Markov modulated diffusion process. Relying on the semimartingale representation of the Markov chain we stated the corresponding HJB equations. We found explicit solutions for the optimal control and derived the value function up to an expectation over the Markov chain. Furthermore, we confirmed the practicability of the derived formulas by some numerical implementations. We found out that the optimal portfolio consists of two parts: the first one corresponds to the solution of the mean-variance optimization problem and the second one corrects it for the additional risk coming from the stochastic factor. What is more, we showed that the state of the Markov chain influences strongly the optimal investment decision. Besides these practical insights we proved an easy to check verification theorem that reduces the case with Markov switching to the one with time-dependent coefficients. We applied it to verify the solution for the Markov modulated Heston model. To sum up, we contributed to literature by solving the optimal investment problem in a realistic model and analyzing the results from theoretical and practical point of view.

\appendix
\section{General mathematical results on Markov chains and Markov modulated It\^o diffusions}
\label{SecPrel}
In this section we give an overview over some properties of Markov chains that will be of use for the solution of the considered optimization problem. The longer proofs are outsourced in the appendix.\\
The next theorem states that a Markov chain is a semimartingale and reveals from the point of view of semimartingale theory the importance of the intensity matrix. For a proof see \cite{Yin1998}, p.18, Lemma 2.4.
\begin{theorem}[Semimartingale decomposition]$\mbox{}$\\ 
Let $\MC$ be a Markov chain with state space $\E=\{e_1,\ldots,e_l\}$ and generator matrix $Q$. Then, $\MC$ can be written in the following form:
\begin{align}
\MC(t)=\MC(0)+\int_{0}^{t}Q^{\prime}\MC(s)\text{d}s+M(t),\forall t\in[0,\infty),
\label{semimart}
\end{align}
where $M$ is a martingale w.r.t. the filtration $\IF^{\MC}=\{\F_t^{\MC}\}_{t\geq0}$ generated by process $\MC$. Observe that process $\int_{0}^{t}Q^{\prime}\MC(s)\text{d}s$ is of finite variation. Thus, Equation (\ref{semimart}) is a semimartingale representation of $\MC$.
\label{ThSemimart}
\end{theorem}
A key result needed for our derivations is the It\^o's formula for Markov modulated It\^o diffusions. Before stating the main result, let us define formally what we call a Markov modulated It\^o diffusion.
\begin{Def}[Markov modulated It\^o diffusion]$\mbox{}$\\ 
Let $W:=\{W(t)\}_{t\geq 0}$ be an $m$-dimensional Brownian motion and $\MC:=\{\MC(t)\}_{t\geq 0}$ as in Theorem \ref{ThSemimart}.
Then the process $X=\{X(t)\}_{t\geq 0}=\big\{\big(X_1(t),\ldots,X_n(t)\big)^{\prime}\big\}_{t\geq 0}\in\IR^n$ defined via:
\begin{align}
X(t)&=X(0)+\int_0^t\mu_X\big(s,X(s),\MC(s)\big)\text{d}s+\int_0^t\sigma_X\big(s,X(s),\MC(s)\big)\text{d}W(s)\\
&=\left(\begin{array}{c}x^1_0\\\vdots\\x^n_0\end{array}\right)+\int_0^t\left(\begin{array}{c}\mu^1_X\big(s,X(s),\MC(s)\big)\\\vdots\\\mu^n_X\big(s,X(s),\MC(s)\big)\end{array}\right)\text{d}s\notag\\
&+\int_0^t\left(\begin{array}{ccc}\sigma^{11}_X\big(s,X(s),\MC(s)\big)&\ldots&\sigma^{1m}_X\big(s,X(s),\MC(s)\big)\\\vdots&\vdots&\vdots\\\sigma^{n1}_X\big(s,X(s),\MC(s)\big)&\ldots&\sigma^{nm}_X\big(s,X(s),\MC(s)\big)\end{array}\right)\text{d}\left(\begin{array}{c}W_1(s)\\\vdots\\W_m(s)\end{array}\right)\notag
\end{align}
is called a Markov modulated It\^o diffusion, where $\mu_X:[0,\infty)\times \IR^n\times\E\rightarrow\IR^n$ and $\sigma_X:[0,\infty)\times \IR^n\times\E\rightarrow\IR^{n,m}$ are deterministic functions. Observe that $X$ is a continuous process as we have two integrals w.r.t. continuous processes. The natural filtration of $X$ is denoted by $\IF^{X}=\{\F_t^X\}_{t\geq0}$.
\label{DefMarkModDiff}
\end{Def}
\begin{theorem}[It\^o's formula for Markov modulated It\^o diffusions]$\mbox{}$\\ 
Consider process $X$ as in Definition \ref{DefMarkModDiff}. Further, a function $f:[0,\infty)\times \IR^n\times\E\rightarrow\IR, \big(t,(x_1,\ldots,x_n)^{\prime},e_i\big)\mapsto f\big(t,(x_1,\ldots,x_n)^{\prime},e_i\big)$ is given, which is once continuously differentiable in the first variable and twice continuously differentiable in the second for all $e_i\in\E$. Then $\big\{f\big(t,X(t),\MC(t)\big)\big\}_{t\geq0}$ is a semimartingale and we have:
\begin{align*}
\begin{aligned}
&f\big(t,X(t),\MC(t)\big)=f(0,X(0),\MC(0))\\
&+\int_0^t\Big[\frac{\partial}{\partial t}f\big(s,X(s),\MC(s)\big)+\sum_{j=1}^{n}\frac{\partial}{\partial x_j}f\big(s,X(s),\MC(s)\big)\mu^j_x\big(s,X(s),\MC(s)\big)\\
&+\frac{1}{2}\sum_{j,k=1}^{n}\frac{\partial^2}{\partial x_j\partial x_k}f\big(s,X(s),\MC(s)\big)\sum_{r=1}^{m}\sigma_X^{jr}\big(s,X(s),\MC(s)\big)\sigma_X^{kr}\big(s,X(s),\MC(s)\big)\Big]\text{d}s\\
&+\int_0^t\sum_{r=1}^{m}\Big(\sum_{j=1}^{n}\sigma^{jr}_X\big(s,X(s),\MC(s)\big)\frac{\partial}{\partial x_j}f\big(s,X(s),\MC(s)\big)\Big)\text{d}W_r(s)\\
&+\int_0^t\sum_{i=1}^{l}f\big(s,X(s),e_i\big)q_{MC(s)i}\text{d}s+\int_0^t\sum_{i=1}^{l}f(s,X(s),e_i)\text{d}M_i(s).
\end{aligned}
\end{align*}
\label{ThIto}
\end{theorem}
\begin{proof}
Follows as an application of the general It\^o formula for semimartingales from \cite{Jacod2003}, p.57, Theorem 4.47.
\end{proof}
Another important result about  It\^o diffusions extended by Markov chains is the connection between the solution of a deterministic PDE and expectations of exponentials. A generalization of the classical Feynman-Kac Theorem is presented in the following theorem.
\begin{theorem}[Feynman-Kac Theorem with Markov switching I]$\mbox{}$\\
Let process $X$, valued in the set $D_X\subseteq\IR$, be a one-dimensional Markov modulated  It\^o diffusion given by the following SDE:
$$
\text{d}X(t)={\mu}_X\big(X(t),\MC(t)\big)\text{d}t+{\sigma}_X\big(X(t), \MC(t)\big)\text{d}W(s),
$$
where $W$ is a Brownian motion and $\MC$ is the Markov chain from Theorem \ref{ThSemimart}. 
Further consider function $K:[0,T]\times D_X\times \E\rightarrow\IR$. Then, define function $k:[0,T]\times D_X\times\E\rightarrow\IR$ via:
$$k(t,x,e_i):=\IE\Big[\exp\Big\{-\int_0^tK\big(t-s,X(s),\MC(s)\big)\text{d}s\Big\}\Big|X(0)=x,\MC(0)=e_i\Big].$$
For all $\ (t,x,e_i)\in[0,T]\times D_X\times\E$, assume that function $k$ is well-defined, $k(t,x,e_i)<\infty$,  and that the following conditions hold:
\begin{enumerate}
\renewcommand{\labelenumi}{\roman{enumi})}
\item Function $k$ is twice continuously differentiable in $x$.
\item For the process $\{N(r)\}_{r\in[0,T]}$ defined by:
\begin{align*}
&N(r):=\frac{1}{r}\int_0^r\Big[\frac{\partial}{\partial x} k\big(t,X(s),\MC(s)\big){\mu}_X\big(X(s),\MC(s)\big)\\
&+\frac{1}{2}\frac{\partial^2}{\partial x\partial x}k\big(t,X(s),\MC(s)\big)\sigma_X^2\big(X(s),\MC(s)\big)+\sum_{j=1}^{l}q_{MC(s)j}k(t,X(s),e_j)\Big]\text{d}s,
\end{align*}
it holds that:
$$\lim_{r\downarrow0}\IE[N(r)|X(0)=x,\MC(0)=e_i]=\IE[\lim_{r\downarrow0}N(r)|X(0)=x,\MC(0)=e_i].$$ This means, it converges in mean to its almost sure limit, for all $t\in[0,T]$.
\item $\IE\Big[\int_0^r\frac{\partial}{\partial x} k\big(t,X(s),\MC(s)\big)\sigma_X\big(X(s),\MC(s)\big)\text{d}W(s)\Big|X(0)=x,\MC(0)=e_i\Big]=0,$ for all $r\in[0,T]$.
\item $\IE\Big[\int_0^r\big(k(t,X(s),e_1),\ldots,k(t,X(s),e_l)\big)\text{d}M(s)\Big|X(0)=x,\MC(0)=e_i\Big]=0,$
for all $r\in[0,T]$.
\item For
\begin{align*}
Z(t+r):&=\exp\Big\{-\int_0^{t+r}K\big(t+r-s,X(s),\MC(s)\big)\text{d}s\Big\}\\
Y(r):&=\exp\Big\{\int_0^{r}K\big(t+r-s,X(s),\MC(s)\big)\text{d}s\Big\},
\end{align*}
it holds, for all $r\in[0,T-t]$, that: 
\begin{align*}
&\lim_{r\downarrow0}\IE\Big[Z(t+r)\frac{Y(r)-Y(0)}{r}\Big|X(0)=x,\MC(0)=e_i\Big]\\
&=\IE\Big[\lim_{r\downarrow0}Z(t+r)\frac{Y(r)-Y(0)}{r})\Big|X(0)=x,\MC(0)=e_i\Big]\\
&=\IE[Z(t)\Big|X(0)=x,\MC(0)=e_i]K(t,x,e_i).
\end{align*}
\end{enumerate}
Then $k$ is differentiable w.r.t. $t$ and satisfies the following system of coupled PDEs, for all $(t,x)\in[0,T]\times D_X$:
\begin{align}
\begin{aligned}
&-\frac{\partial}{\partial t}k(t,x,e_i)-k(t,x,e_i)K(t,x,e_i)+\frac{\partial}{\partial x}k(t,x,e_i){\mu}_X(x,e_i)\\
&+\frac{1}{2}\frac{\partial^2}{\partial x\partial x}k(t,x,e_i)\sigma_X^2(x,e_i)=-\sum_{j=1}^lq_{ij}k(t,x,e_j),\ \  k(0,x,e_i)=1,\forall i\in E.
\end{aligned}
\label{PDEThFeynmKac}
\end{align}
Further, conditions iii) and iv) and v) from above can be replaced by the following conditions, respectively:
\begin{enumerate}
\item[iii)'] 
$\IE\Big[\int_0^r\Big\{\frac{\partial}{\partial x} k\big(t,X(s),\MC(s)\big)\sigma_X\big(X(s),\MC(s)\big)\Big\}^2\text{d}s\Big|X(0)=x,\MC(0)=e_i\Big]<\infty$,  for all $r\in[0,T]$.
\item[iv)']$\IE\Big[\int_0^r\big(k(t,X(s),e_j)\big)^2\text{d}s\Big|X(0)=x,\MC(0)=e_i\Big]<\infty,$
for all $r\in[0,T]$.
\item[v)'] $K(t,x,e_i)$ and $\frac{\partial}{\partial t}K(t,x,e_i)$ are continuous in $t$ and $x$ and
\begin{align*}
&\lim_{r\downarrow0}\IE\Big[Z(t+r)\frac{Y(r)-Y(0)}{r}\Big|X(0)=x,\MC(0)=e_i\Big]\\
&=\IE\Big[\lim_{r\downarrow0}Z(t+r)\frac{Y(r)-Y(0)}{r})\Big|X(0)=x,\MC(0)=e_i\Big],
\end{align*}
for all $r\in[0,T-t]$, $e_j\in\E$. Note that a sufficient condition for the equality above to hold is the existence of an integrable bound for $Z(t+r)\frac{Y(r)-Y(0)}{r}$.
\end{enumerate}
\label{ThFeynmKac}
\end{theorem}
\begin{proof}
The proof goes along the lines of the proof of Theorem 8.2.1 from \cite{Oksendal2000}. 
\end{proof}
\begin{remark}
Some boundedness conditions on $k$, ${\mu}_X$, $\sigma_X$ and $K$ can easily replace assumptions (ii), (iii), (iv) and (v). However these functions are not bounded for important examples, such as the Heston model. That is why we keep the assumptions as general as possible.
\end{remark}
For our applications we need the backwards formulation of the Feynman-Kac result, which follows as a direct application of the theorem above and is stated in the following corollary:
\begin{corollary}[Feynman-Kac Theorem with Markov switching II]$\mbox{}$\\ 
Consider processes $X$ and $\MC$ as in Theorem \ref{ThFeynmKac} and some function $H:[0,T]\times D_X\times \E\rightarrow\IR$. Define function $h:[0,T]\times D_X\times\E$ via:
$$h(t,x,e_i):=\IE\Big[\exp\Big\{-\int_t^TH\big(s,X(s),\MC(s)\big)\Big\}\Big|X(t)=x,\MC(t)=e_i\Big].$$
Denote $K(t,x,e_i):=H(T-t,x,e_i)$ and define function $k(t,x,e_i)$ as in Theorem \ref{ThFeynmKac}. Assume that the conditions of Theorem \ref{ThFeynmKac} hold for $k$ and $K$. Then $h$ is differentiable w.r.t. $t$ and satisfies the following system of coupled PDEs for all $(t,x)\in[0,T]\times D_X$:
\begin{align*}
&\frac{\partial}{\partial t}h(t,x,e_i)-h(t,x,e_i)H(t,x,e_i)+\frac{\partial}{\partial x}h(t,x,e_i){\mu}_X(x,e_i)\\
&+\frac{1}{2}\frac{\partial^2}{\partial x\partial x}h(t,x,e_i)\sigma_X^2(x,e_i)=-\sum_{j=1}^lq_{ij}h(t,x,e_j), \ \ h(T,x,e_i)=1, \forall i \in E.
\end{align*}
\label{CorFK}
\end{corollary}
As an application of the backwards Feynman-Kac result one obtains the following corollary expectations over Markov chains and their relation to linear systems of coupled ODEs:
\begin{corollary}[Linear system of coupled ODEs]$\mbox{}$\\ 
Consider again the Markov chain $\MC$ from Theorem \ref{ThSemimart} and define function $\xi:[0,T]\times\E\rightarrow\IR$ as follows:
\begin{align*}
\xi(t,e_i)=\IE\Big[\exp\Big\{-\int_t^T\Xi\big(s,\MC(s)\big)\text{d}s\Big\}\Big|\MC(t)=e_i\Big],
\end{align*}
for some finite function $\Xi:[0,T]\times\E\rightarrow\IR$. Assume that function $\xi$ is well-defined and $\xi(t,e_i)<\infty,\forall\ (t,e_i)\in[0,T]\times\E$. Further, for an arbitrary but fix $t\in[0,T]$ and all $r\in[0,t]$, define:
\begin{align*}
Z(t+r):&=\exp\Big\{-\int_0^{t+r}\Xi\big(T-t-r+s,\MC(s)\big)\text{d}s\Big\}\\
Y(r):&=\exp\Big\{\int_0^{r}\Xi\big(T-t-r+s,\MC(s)\big)\text{d}s\Big\},
\end{align*}
and assume, for all $e_i\in \E$, that 
\begin{align}
\begin{aligned}
&\lim_{r\downarrow0}\IE\Big[Z(t+r)\frac{Y(r)-Y(0)}{r}\Big|\MC(0)=e_i\Big]\\
&=\IE\Big[\lim_{r\downarrow0}Z(t+r)\frac{Y(r)-Y(0)}{r})\Big|\MC(0)=e_i\Big]\\
&=\IE[Z(t)\Big|\MC(0)=e_i]\Xi(T-t,e_i).
\end{aligned}
\label{ConCorrXi}
\end{align}
Alternatively to assuming (\ref{ConCorrXi}) one can require that $\Xi(t,e_i)$ and $\frac{\partial}{\partial t}\Xi(t,e_i)$ are continuous in $t$ for all $e_i\in\E$.\\
Then $\xi$ satisfies the following system of ODEs for all $t\in[0,T]$:
\begin{align*}
&\frac{\partial}{\partial t}\xi(t,e_i)-\xi(t,e_i)\Xi(t,e_i)=-\sum_{j=1}^lq_{ij}\xi(t,e_j), \ \ \xi(T,e_i)=1, \forall i\in E.
\end{align*}
\label{CorFKLS}
\end{corollary}
\section{Appendix for Section \ref{SecSol}}
\label{AppSecSol}
\numberwithin{equation}{section}
\begin{proof}[Proof of Theorem \ref{ThVerifTimeDGen}]$\mbox{}$\\
Consider an arbitrary but fixed point $(t,v)\in[0,T]\times\IR_{\geq0}$ and an admissible strategy $\pi\in\Lambda^m(t,v)$ and apply It\^o's rule stepwise to $\Phi^m\big(T,V^{m,\pi}(T),X^m(T)\big)$:
\begin{align*}
&\Phi^m\big(T,V^{m,\pi}(T),X^m(T)\big)=\Phi^m\big(t_K,V^{m,\pi}(t_K),X^m(t_K)\big)+\int_{t_K}^{T}\pazocal{L}^m(m_K,\pi)\Phi^m\big(s,V^{m,\pi}(s),X^m(s)\big)\text{d}s\displaybreak[0]\\
&+\int_{t_K}^{T}\Phi^m_v\big(s,V^{m,\pi}(s),X^m(s)\big)\sigma_V\big(V^{m,\pi}(s),X^m(s),m(s),\pi(s)\big)\text{d}W_P(s)\displaybreak[0]\\
&+\int_{t_K}^{T}\Phi^m_x\big(s,V^{m,\pi}(s),X^m(s)\big)\sigma_X\big(X^m(s),m(s)\big)\text{d}W_X(s)\displaybreak[0]\\
&=\Phi^m\big(t_{K-1},V^{m,\pi}(t_{K-1}),X^m(t_{K-1})\big)+\sum_{i=K-1}^{K}\int_{t_{i}}^{t_{i+1}}\pazocal{L}^m(m_i,\pi)\Phi^m\big(s,V^{m,\pi}(s),X^m(s)\big)\text{d}s\displaybreak[0]\\
&+\int_{t_{K-1}}^{T}\Phi^m_v\big(s,V^{m,\pi}(s),X^m(s)\big)\sigma_V\big(V^{m,\pi}(s),X^m(s),m(s),\pi(s)\big)\text{d}W_P(s)\displaybreak[0]\\
&+\int_{t_{K-1}}^{T}\Phi^m_x\big(s,V^{m,\pi}(s),X^m(s)\big)\sigma_X\big(X^m(s),m(s)\big)\text{d}W_X(s)\\
&=\ldots\displaybreak[0]\\
&=\Phi^m\big(t,V^{m,\pi}(t),X^m(t)\big)+\sum_{i=0}^{K}\int_{t_i}^{t_{i+1}}\underbrace{\pazocal{L}^m(m_i,\pi)\Phi^m\big(s,V^{m,\pi}(s),X^m(s)\big)}_{\leq 0}\text{d}s\displaybreak[0]\\
&+\int_{t}^{T}\Phi^m_v\big(s,V^{m,\pi}(s),X^m(s)\big)\sigma_V\big(V^{m,\pi}(s),X^m(s),m(s),\pi(s)\big)\text{d}W_P(s)\displaybreak[0]\\
&+\int_{t}^{T}\Phi^m_x\big(s,V^{m,\pi}(s),X^m(s)\big)\sigma_X\big(X^m(s),m(s)\big)\text{d}W_X(s)\displaybreak[0]\\
&\leq \Phi^m\big(t,V^{m,\pi}(t),X^m(t)\big)+\int_{t}^{T}\Phi^m_v\big(s,V^{m,\pi}(s),X^m(s)\big)\sigma_V\big(V^{m,\pi}(s),X^m(s),m(s),\pi(s)\big)\text{d}W_P(s)\displaybreak[0]\\
&+\int_{t}^{T}\Phi^m_x\big(s,V^{m,\pi}(s),X^m(s)\big)\sigma_X\big(X^m(s),m(s)\big)\text{d}W_X(s)=:Y^m(T).
\end{align*}
Note that here we have used the continuity and the piecewise differentiability of function $\Phi^m$. One can show the statement analogously for an arbitrary end-point $\tau\in[t,T]$: 
\begin{align}
\begin{aligned}
&\Phi^m\big(\tau,V^{m,\pi}(\tau),X^m(\tau)\big)\leq \Phi^m\big(t,V^{m,\pi}(t),X^m(t)\big)\\
&+\int_{t}^{\tau}\Phi^m_v\big(s,V^{m,\pi}(s),X^m(s)\big)\sigma_V\big(V^{m,\pi}(s),X^m(s),m(s),\pi(s)\big)\text{d}W_P(s)\\
&+\int_{t}^{\tau}\Phi^m_x\big(s,V^{m,\pi}(s),X^m(s)\big)\sigma_X\big(X^m(s),m(s)\big)\text{d}W_X(s)=:Y^m(\tau).
\end{aligned}
\label{helpSupMart}
\end{align}
Now assume that $\Phi^m\big(\tau, v, x\big)\geq0$ for all $(\tau, v,x)\in[0,T]\times\IR_{\geq0}\times D_X$. It follows that $Y^m(\tau)\geq0$. Furthermore, $Y^m$ is a local martingale, as all involved functions are at least piecewise continuous, so it is a supermartingale. Then it holds that:
\begin{align*}
&\IE\Big[U\big(V^{m,\pi}(T)\big)\Big|\F_t\Big]=\IE\Big[\frac{\big(V^{m,\pi}(T)\big)^{\delta}}{\delta}\Big|\F_t\Big]=\IE\Big[\Phi^m\big(T,V^{m,\pi}(T),X^m(T)\big)\Big|\F_t\Big]\\
&\leq \IE[Y^m(T)|\F_t] \leq Y^m(t) = \Phi^m\big(t,V^{m,\pi}(t),X^m(t)\big),
\end{align*}
which proves the first statement of the theorem.\\
We continue with the proof for the second statement. As $\Phi^m$ is a martingale, it follows directly:
\begin{align*}
&\IE\Big[\frac{\big(V^{m,\bar{\pi}^m}(T)\big)^{\delta}}{\delta}\Big|\F_t\Big]=\IE\Big[\Phi^m\big(T,V^{m,\bar{\pi}^m}(T),X^m(T)\big)\Big|\F_t\Big]=\Phi^m\big(t,V^{m,\bar{\pi}^m}(t),X^m(t)\big).
\end{align*}
Our proof is complete.
$\mbox{}$\halmos
\end{proof}
\begin{remark}
Observe that we have not used the exponential structure of our model for this proof. Thus, the result holds for general time-dependent models with a stochastic factor.
\end{remark}
\begin{lemma}[Exponential affine martingales]$\mbox{}$\\
Let $Z=(Z_1,\ldots,Z_n)^{\prime}$ be an n-dim continuous semimartingale with affine differential characteristics\footnote{For a definition of differential characteristics for semimartingales see \cite{Kallsen2010}.} $(\mu_Z,\gamma_Z)\in\IR^n\times \IR^{n\times n}$, i.e.
\begin{align*}
\mu_Z(t)&=\left(\begin{array}{c}\mu_Z^1(t)\\\vdots\\\mu_Z^n(t)\end{array}\right)=\left(\begin{array}{c}\alpha_0^1(t)\\\vdots\\\alpha_0^n(t)\end{array}\right)+\sum_{j=1}^{n}Z_j\left(\begin{array}{c}\alpha_j^1(t)\\\vdots\\\alpha_j^n(t)\end{array}\right)\\
\gamma_Z(t)&=\left(\begin{array}{ccc}\gamma_Z^{11}(t)&\ldots&\gamma_Z^{1n}(t)\\\vdots&\vdots&\vdots\\\gamma_Z^{n1}(t)&\ldots&\gamma_Z^{nn}(t)\end{array}\right)=\left(\begin{array}{ccc}\beta_0^{11}(t)&\ldots&\beta_0^{1n}(t)\\\vdots&\vdots&\vdots\\\beta_0^{n1}(t)&\ldots&\beta_0^{nn}(t)\end{array}\right)+\sum_{j=1}^{n}Z_j\left(\begin{array}{ccc}\beta_j^{11}(t)&\ldots&\beta_j^{1n}(t)\\\vdots&\vdots&\vdots\\\beta_j^{n1}(t)&\ldots&\beta_j^{nn}(t)\end{array}\right),
\end{align*}
for some deterministic functions $\alpha_j^{k},\beta_j^{k,l}:[0,\infty]\rightarrow\IR$, $j\in\{0,\ldots,n\}$, $k,l\in\{1,\ldots,n\}$. Further assume that there exists a number $p\in\IN$, $p\leq n$ such that for all $t\in\IR_{\geq0}$:
\begin{enumerate}
\renewcommand{\labelenumi}{\roman{enumi})}
\item $\beta_{j}^{kl}(t)=0$ if $0\leq j\leq p$, $1\leq k,l\leq p$ unless $k=l=j$;
\item $\alpha_{j}^{k}(t)=0$ if $j\geq p+1$, $1\leq k\leq p$;
\item $\beta_{j}^{kl}(t)=0$ if $j\geq p+1$, $1\leq k,l\leq n$.
\end{enumerate}
If additionally $\alpha_j(t)$ and $\beta_j(t)$ are continuous in $t\in\IR_{\geq0}$ for all $0\leq j\leq n$ and the following condition holds for some $1\leq i\leq n$:
\begin{align}
\alpha_j^{i}(t)+\frac{1}{2}\beta^{ii}_j(t)=0, \forall\ 0\leq j\leq n,
\label{condLocMart}
\end{align}
then $\big\{\exp\{Z_i(t)\}\big\}_{t\in[0,\infty)}$ is a martingale. 
\label{LemmaKallsen}
\end{lemma}
\begin{proof}[Proof for Theorem \ref{ThVerifTimeD}]$\mbox{}$\\
 As function $\Phi^m$ given by Equation (\ref{EqSolTimeD}) is obviously continuous and piecewise sufficiently differentiable, by Theorem \ref{ThVerifTimeDGen} we only need to show that $\big\{\Phi^{m}\big(t,V^{m,\bar{\pi}^{m}}(t),X^m(t)\big)\}_{t\in[0,T]}$ is a martingale. We start by writing down the solution of the SDE for $V^{m,\bar{\pi}^{m}}$:
\begin{align*}
V^{m,\bar{\pi}^{m}}(t)=&v_0\exp\Big\{\int_0^tr\big(m(s)\big)+\lambda\big(X^m(s),m(s)\big)\bar{\pi}^{m}(s)\\
&-\frac{1}{2}\big(\bar{\pi}^{m}(s)\big)^2\sigma_P^2\big(X^m(s),m(s)\big)\text{d}s+\int_0^t\bar{\pi}^{m}(s)\sigma_P\big(X^m(s),m(s)\big)\text{d}W_P(s)\Big\},
\end{align*}
for all $0\leq t\leq T$. Then we insert it in the expression for $\Phi^m$:
\begin{align*}
&\Phi^{m}\big(t,V^{m,\bar{\pi}^{m}}(t),X^m(t)\big)=\frac{\big(V^{m,\bar{\pi}^{m}}(t)\big)^{\delta}}{\delta}\exp\Big\{\int_t^T\delta r\big(m(s)\big)\text{d}s\Big\}\exp\big\{\vartheta A^m(t)\\
&+\vartheta B^m(t)X^m(t)\big\}=\underbrace{\frac{v_0^{\delta}}{\delta}\exp\Big\{\int_0^T\delta r\big(m(s)\big)\text{d}s\Big\}\exp\{\vartheta A^m(0)+\vartheta B^m(0)x_0\}}_{=\Phi^{m}(0,v_0,x_0)}\\
&\cdot\exp\Big\{\int_0^t\delta \lambda\big(X^m(s),m(s)\big)\bar{\pi}^{m}(s)-\frac{1}{2}\delta \big(\bar{\pi}^{m}(s)\big)^2\sigma_P^2\big(X^m(s),m(s)\big)\\
&+\vartheta A_t^m(s)+\vartheta B^m_t(s)X^m(s)+\vartheta B^m(s)\mu_X\big(X^m(s),m(s)\big)\text{d}s\\
&+\int_0^t\underbrace{\vartheta B^m(s)\sigma_X\big(X^m(s),m(s)\big)}_{=:\sigma_L^{(1)}(s,X^m(s))}\text{d}W_X(s)\Big\}+\int_0^t\underbrace{\delta \bar{\pi}^{m}(s)\sigma_P\big(X^m(s),m(s)\big)}_{=:\sigma_L^{(2)}(s,X^m(s))}\text{d}W_P(s)\\
&=:\Phi^{m}(0,v_0,x_0)\exp\Big\{\int_0^t\mu_L\big(s,X^m(s)\big)\text{d}s+\int_0^t\sigma_L^{(1)}\big(s,X^m(s)\big)\text{d}W_X(s)\\
&+\int_0^t\sigma_L^{(2)}\big(s,X^m(s)\big)\text{d}W_P(s)\Big\}=:\Phi^{m}(0,v_0,x_0)\exp\{L^m(t)\},
\end{align*}
Now we recognize easily the differential semimartingale characteristics $\mu_Z(t)=\left(\begin{array}{c}\mu_Z^{(1)}\big(t,X^m(t)\big)\\\mu_Z^{(2)}\big(t,X^m(t)\big)\end{array}\right)$ and $\Gamma_Z(t)=\left(\begin{array}{cc}\Gamma_Z^{(11)}\big(t,X^m(t)\big)&\Gamma_Z^{(12)}\big(t,X^m(t)\big)\\\Gamma_Z^{(21)}\big(t,X^m(t)\big)&\Gamma_Z^{(22)}\big(t,X^m(t)\big)\end{array}\right)$ of the two dimensional process $Z:=(X^m,L^m,)^{\prime}$:
\begin{align*}
\mu_Z^{(1)}\big(t,x\big)=&\mu_X(x,m(t))=\mu^{(1)}_X\big(m(t)\big)+x\mu^{(2)}_X\big(m(t)\big)\\
\mu_Z^{(2)}\big(t,x\big)=&\mu_L(t,x)=\delta \lambda\big(x,m(t)\big)\bar{\pi}^{m}(t)-\frac{1}{2}\delta \big(\bar{\pi}^{m}(t)\big)^2\sigma_P^2\big(x,m(t)\big)\\
&+\vartheta A_t^m+\vartheta B^m_tx+\vartheta B^m\mu_X\big(x,m(t)\big)\\
=&\vartheta A^m_t+\vartheta B^m\mu_X^{(1)}+x\Big\{\vartheta B^m_t+\vartheta B^m\mu_X^{(2)}+\frac{\delta}{1-\delta}\Gamma^{(2)}\big(1-\frac{1}{2(1-\delta)}\big)\\
&-\vartheta B^m\zeta^{(2)}\frac{\delta^2}{(1-\delta)^2}-\frac{1}{2}\frac{\delta}{(1-\delta)^2}\vartheta^2(B^m)^2\rho^2\Sigma_X^{(2)}\Big\}\\
\Gamma_Z^{(11)}(t,x)=&\sigma_X^2(x,m(t))=x\Sigma_X^{(2)}\big(m(t)\big)\\
\Gamma_Z^{(12)}(t,x)=&\Gamma_Z^{(21)}(t,x)=\sigma_X(x,m(t))\big(\sigma_L^{(1)}(t,x)+\rho\sigma_L^{(2)}(t,x)\big)\\
&=x\Big\{B^m(t) \Sigma_X^{(2)}\big(m(t)\big)+\frac{\delta}{1-\delta}\zeta^{(2)}\big(m(t)\big)\Big\}\\
\Gamma_Z^{(22)}(t,x)=&\big(\sigma_L^{(1)}(t,x)\big)^2+\big(\sigma_L^{(2)}(t,x)\big)^2+2\rho\sigma_L^{(1)}(t,x)\sigma_L^{(2)}(t,x)\\
=&x\Big\{\vartheta^2(B^m(t))^2\Sigma_X^{(2)}\big(m(t)\big)\Big(1+\frac{\rho^2\delta^2}{(1-\delta)^2}+2\frac{\rho^2\delta}{1-\delta}\Big)\\
&+2\vartheta B^m(t)\zeta^{(2)}\big(m(t)\big)\frac{\delta}{(1-\delta)^2}+\frac{\delta^2}{(1-\delta)^2}\Gamma^{(2)}\big(m(t)\big)\Big\},\\
\end{align*}
where we have omitted the dependence on $m$, $t$ and $x$ for reasons of better readability and we have substituted the following equalities:
\begin{align*}
\bar{\pi}^{m}(t)&=\frac{1}{1-\delta}\Big\{\frac{\lambda\big(x,m(t)\big)}{\sigma^2_P\big(x,m(t)\big)}+\rho\frac{\sigma_X\big(x,m(t)\big)}{\sigma_P\big(x,m(t)\big)}\vartheta B^m(t)\Big\}\\
\sigma^2_X\big(x,m(t)\big)&=\Sigma_X^{(2)}\big(m(t)\big)x\\
\mu_X\big(x,m(t)\big)&=\mu_X^{(1)}\big(m(t)\big)+\mu_X^{(2)}\big(m(t)\big)x\\
\rho\lambda\big(x,m(t)\big)\frac{\sigma_X\big(x,m(t)\big)}{\sigma_P\big(x,m(t)\big)}&=\zeta^{(2)}\big(m(t)\big)x\\
\left(\frac{\lambda\big(x,m(t)\big)}{\sigma_P\big(x,m(t)\big)}\right)^2&=\Gamma^{(2)}\big(m(t)\big)x.
\end{align*}
Observe that $\mu_Z$ and $\Gamma_Z$ are piecewise continuous and fulfill conditions i),ii),iii) from Lemma \ref{LemmaKallsen} with $p=1$. Next we show that $\mu^{(2)}_Z(t,x)+\frac{1}{2}\Gamma^{(22)}_Z(t,x)=0$, which by comparison of coefficients is equivalent to Equation (\ref{condLocMart}) for $i=2$:
\begin{align*}
\mu^{(2)}_Z(t,x)+\frac{1}{2}\Gamma^{(22)}_Z(t,x)
=&\vartheta\Big[A^m_t+\mu_X^{(1)}B^m\Big]+\vartheta x \Big[B^m_t+\frac{1}{2}\Sigma_X^{(2)}\big(B^m\big)^2\\
&+\big\{\frac{\delta}{1-\delta}\zeta^{(2)}+\mu_X^{(2)}\big\}B^m+\frac{1}{2}\frac{1}{\vartheta}\frac{\delta}{1-\delta}\Gamma^{(2)}\Big]=0,
\end{align*}
where the last equality holds because of Equations (\ref{Riccati}) and (\ref{Integr}). It follows from Lemma \ref{LemmaKallsen} that process $\{\exp\{L^m(t)\}\}_{t\in[0,T]}$ and thus also process $\{\Phi^{m}\big(t,V^{m,\bar{\pi}^{m}}(t),X^m(t)\big)\}_{t\in[0,T]}$ are martingales. Application of Theorem \ref{ThVerifTimeDGen} completes the proof.\\
$\mbox{}$\halmos
\end{proof}
\addcontentsline{toc}{section}{Bibliography}
\bibliographystyle{abbrvnat}
\bibliography{PortfOptimAffineMS}
\end{document}